\documentclass[10pt]{article}
\usepackage{amsmath}
\usepackage{amssymb}
\usepackage{enumerate}
\usepackage{epsf}
\usepackage{psfrag}
\usepackage{mathrsfs}
\usepackage{hyperref}
\usepackage{slashed}
\usepackage{mathrsfs}
\usepackage{amsthm} 
\usepackage{tensor} 
\usepackage{mathtools} 
\usepackage{cancel} 
\usepackage{empheq}

\usepackage{tabulary} 
\usepackage{float} 

\usepackage[table]{xcolor} 
\usepackage{booktabs}
\renewcommand{\arraystretch}{1.3}
\DeclareGraphicsExtensions{.eps,.art,.ART,.ps}
\newcommand{\bbR}{\mathbb{R}}      

\makeatletter
\def\pdv{\@ifnextchar[{\@pdvwith}{\@pdvwithout}}
\def\@pdvwith[#1]#2{\frac{\partial^{#1}}{\partial {#2}^{#1}}}
\def\@pdvwithout#1{\frac{\partial}{\partial {#1}}}
\def\dv{\@ifnextchar[{\@dvwith}{\@dvwithout}}
\def\@dvwith[#1]#2{\frac{d^{#1}}{d {#2}^{#1}}}
\def\@dvwithout#1{\frac{d}{d {#1}}}
\makeatother
\newcommand{\disth}{\text{dist}_{\mathbb{H}^3}} 
\newcommand{\dists}{\text{dist}_{S^3}} 
\DeclareMathOperator{\cotan}{cotan}
\DeclareMathOperator{\cotanh}{cotanh}
\DeclareMathOperator{\sech}{sech}
\DeclareMathOperator{\cosech}{cosech}

\DeclareMathOperator{\arctanh}{arctanh}

\newtheorem{Thm}{Theorem}[section]

\newtheorem{Conjecture}[Thm]{Conjecture}
\theoremstyle{definition}

\begin{document}
\title{Explicit formulas and decay rates for the solution of the wave equation in cosmological spacetimes}
\author{Jos\'e Nat\'ario$^{1}$ and Flavio Rossetti$^{1, 2}$\\ \\
{\small $^1$ CAMGSD, Departamento de Matem\'{a}tica, Instituto Superior T\'{e}cnico,}\\
{\small Universidade de Lisboa, Portugal}\\
{\small $^2$ Department of Physics, LMU Munich, Germany}
}
\date{}
\maketitle
\begin{abstract}
We obtain explicit formulas for the solution of the wave equation in certain Friedmann-\allowbreak Lema\^{i}tre-\allowbreak Robertson-\allowbreak Walker (FLRW) spacetimes. Our method, pioneered by Klainerman and Sarnak, consists in finding differential operators that map solutions of the wave equation in these FLRW spacetimes to solutions of the conformally invariant wave equation in simpler, ultra-static spacetimes, for which spherical mean formulas are available. In addition to recovering the formulas for the dust-filled flat and hyperbolic FLRW spacetimes originally derived by Klainerman and Sarnak, and generalizing them to the spherical case, we obtain new formulas for the radiation-filled FLRW spacetimes and also for the de Sitter, anti-de Sitter and Milne universes. We use these formulas to study the solutions with respect to the Huygens principle and the decay rates, and to formulate conjectures about the general decay rates in flat and hyperbolic FLRW spacetimes. The positive resolution of the conjecture in the flat case is seen to follow from known results in the literature.
\end{abstract}
\tableofcontents
%
%
%
\section{Introduction}\label{section0}
The aim of this article is to obtain explicit formulas and exact decay rates for the solution of the wave equation in certain Friedmann-\allowbreak Lema\^{i}tre-\allowbreak Robertson-\allowbreak Walker (FLRW) spacetimes, taken as fixed backgrounds. Besides its intrinsic interest in modelling electromagnetic or gravitational waves, the linear wave equation may be considered as a first step to understand the qualitative behavior of solutions of the Einstein equations, which usually requires a good quantitative grasp of the linearized Einstein equations. 

Explicit formulas, albeit applying only to particular FLRW models, are very useful in that they provide detailed information about the solutions, and may indeed inspire conjectures about the behavior of solutions of the wave equation in other spacetimes (see section~\ref{section6}). Moreover, they highlight special properties of the solutions of the wave equation in the particular FLRW models where they do apply, such as (some form of) the Huygens principle.

The wave equation in cosmological spacetimes has been amply studied in the literature, see for example \cite{AbbasiCraig, CostaNatarioOliveira, DafermosRodnianski1, Gajic, GalstianKinoshitaYagdjian, GalstianYagdjian, GiraoNatarioSilva, KlainermanSarnak, Ringstrom1, Ringstrom2, Ringstrom3, Ringstrom4, Schlue, YagdjianGalstian} and references therein. Our own contribution is described in more detail below.
\subsection{Klainerman-Sarnak method}
The Klainerman-Sarnak method consists of finding an operator $\hat O$ that maps solutions $\phi$ of the wave equation in a given FLRW background to solutions $\hat O \phi$ of the conformally invariant wave equation on a simpler, ultra-static spacetime (characterized by the same spatial curvature), where a spherical means formula for $\hat O \phi$ is available; by inverting the operator $\hat O$, one then obtains an explicit expression for $\phi$. It should be noted, however, that such operators can only be found for very specific FLRW spacetimes.

 Klainerman and Sarnak introduced this method in \cite{KlainermanSarnak} for the dust-filled FLRW universes with flat and hyperbolic spatial sections. Abbasi and Craig extended their analysis in \cite{AbbasiCraig}, where they elaborated on the decay rates of the solutions of the wave equation and also on the validity of the Huygens principle. Physical applications were later discussed in \cite{CraigReddy, StarkoCraig}.

In this paper, we generalize the Klainerman-Sarnak method to the spherical dust-filled FLRW universe, as well as to a number of other FLRW spacetimes. More precisely, using the parameter $K$ to denote the curvature of the spatial sections (so that $K = 0, -1, 1$ corresponds to flat, hyperbolic or spherical geometries, respectively), we obtain new explicit formulas for the solution of the wave equation in radiation-filled ($K = 0, -1, 1$), de Sitter ($K = 0, -1, 1$), anti-de Sitter ($K=-1$) and Milne ($K=-1$) universes. In each case, we are able to find an operator $\hat O$ such that $\hat O \phi$ satisfies the conformally invariant wave equation in:
\begin{itemize}
\item  {\bf Flat case:} Minkowski's spacetime.
\item {\bf Hyperbolic case:} The ultra-static universe with hyperbolic spatial sections (that is, the hyperbolic analogue of Einstein's universe).
\item {\bf Spherical case:} The ultra-static universe with spherical spatial sections (that is, Einstein's universe).
\end{itemize}

 The well-known fact that Minkowski's spacetime is conformal to a globally hyperbolic region of Einstein's universe makes it plausible that a spherical means formula for the conformally invariant wave equation should exist in this universe. A similar argument applies to its hyperbolic analogue, which can be seen to be conformal to a globally hyperbolic region of Minkowski's spacetime.
\subsection{Huygens principle}
All solutions for which we obtain explicit formulas have the property that $\hat O \phi$ satisfies the strong Huygens principle: the values of $\hat O \phi$ at some spacetime point depend only on the values of the initial data at points on its light cone (that is, information propagates along null geodesics). However, as noted by Abbasi and Craig in \cite{AbbasiCraig}, this is not necessarily the case for $\phi$ itself. In this paper, we show that the strong Huygens principle does hold for the radiation-filled ($K = 0, -1, 1$) and the Milne universes. Moreover, a weaker form of the Huygens principle, dubbed the {\em incomplete Huygens principle} in \cite{Yagdjian}, is satisfied by both the de Sitter ($K = 0, -1, 1$) and the anti-de Sitter universes. 
\subsection{Decay rates}
The explicit formulas derived in this paper give very detailed information about the solutions, such as decay rates as $t \to +\infty$, where $t$ is the physical time coordinate of the FLRW metric (see section~\ref{subsection1.5}). 

In the \textbf{flat case}, it is interesting to notice that the decay rates in both the dust-filled and the radiation-filled universes coincide with the decay rate $t^{-1}$ of waves propagating in Minkowski's spacetime, despite the different behavior of their scale factors ($a(t)=t^\frac23$ and $a(t)=t^\frac12$, respectively). These decay rates differ from those obtained in \cite{CostaNatarioOliveira} for the nonzero Fourier modes in the case of \emph{toroidal} spatial sections, which were, respectively, $t^{-\frac{2}{3}}$ and $t^{-\frac{1}{2}}$. The explanation for this discrepancy is dispersion: whereas in the toroidal models the decay is solely due to the cosmological redshift resulting from the expansion, for noncompact spatial sections the solutions can disperse across an unbounded region. The faster the expansion the larger the decay due to the cosmological redshift, but the smaller the decay due to dispersion (as faraway regions are moving away faster). By adapting the analysis of the damped wave equation in \cite{Wirth}, we show that indeed there exists a regime (corresponding to a slow enough expansion of the spatial sections) where the two effects exactly balance each other to give the same decay rate as in Minkowski's spacetime.
For faster expansion, the cosmological redshift does not fully compensate for the lack of dispersion, and the decay becomes slower, until no decay occurs at all.

In the \textbf{hyperbolic case}, we obtain faster decay rates than in the flat case for both the dust-filled and the radiation-filled universes, as could be expected from the fact that dispersion is more effective for hyperbolic spatial geometry. Based on these two examples, and also on the numerical results obtained in \cite{RossettiVinuales}, we formulate a conjecture for the decay rates in general FLRW spacetimes with hyperbolic sections. We also correct a claim in \cite{AbbasiCraig} by showing that the decay rate in the dust-filled universe is $t^{-\frac{3}{2}}$ rather than $t^{-2}$.

\subsection{Behavior at the Big Bang} \label{subsection1.4}
Generically, solutions of the wave equation blow up at the Big Bang (and also at the Big Crunch, when there is one). Nonetheless, some solutions (e.g.~constant functions) have a well-defined limit at the Big Bang. In the cases of the dust-filled ($K = 0, -1, 1$) and radiation-filled ($K = 0, -1, 1$) universes, we show how to obtain explicit expressions for these solutions. In fact, by taking limits of solutions with appropriate initial data, we give explicit expressions for solutions with any given function as its limit at the Big Bang (see also \cite{GiraoNatarioSilva, Ringstrom1}).
\subsection{Summary of the main results and decay conjectures} \label{subsection1.5}
Recall that a $(3+1)$-dimensional FLRW spacetime $(\bbR \times \Sigma_3, g)$ has metric 
\begin{equation} 
g = -dt^2 + a^2(t) d\varSigma_3^{\,2} \, ,
\end{equation}
where $d\varSigma_3^{\,2}$ is the standard Riemannian metric for $\Sigma_3=\bbR^3, \mathbb{H}^3$ or $S^3$, and $a(t)$ is the scale factor. The \emph{conformal} time coordinate is defined as
\begin{equation}
\tau = \int \frac{dt}{a(t)},
\end{equation}
and we refer to $t$ as the \emph{physical} time coordinate. In what follows we give a precise quantitative summary of the main results and conjectures derived from the exact formulas. 

\bigskip

\textbf{Flat case:}
The results for FLRW universes with flat spatial sections are summarized in table~\ref{table:flat}. The decay rates for the dust-filled and the de Sitter universes had already been obtained in \cite{AbbasiCraig} and \cite{NatarioSasane}, respectively.  The fact that the de Sitter universe satisfies the incomplete Huygens principle was first noticed in \cite{Yagdjian}.

\begin{table}[H] 
\begin{center}
\begin{tabulary}{\textwidth}{c|c|c|c|c} 
\toprule  
 Spacetime & Scale factor & $t(\tau)$ & Large times & Huygens \\[.4em]
 \hline 
 \rowcolor{gray!10}
 Minkowski & $a(\tau) = 1$ & $t \propto \tau$ & $|\phi| \lesssim t^{-1}$ & Yes (strong) \\  
 Dust & $a(\tau) \propto \tau^2$ & $t \propto \tau^3$ & $|\phi| \lesssim t^{-1}$ & No \\
 \rowcolor{gray!10}
 Radiation & $a(\tau) \propto \tau$ & $t \propto \tau^2$ & $|\phi| \lesssim t^{-1}$ & Yes (strong)\\
 de Sitter & $a(\tau) \propto \tau^{-1}$ & $t \propto - \log \tau$ & $|\partial_t \phi| \lesssim e^{-2t}$ & Yes (incomplete) \\
\bottomrule
\end{tabulary}
\caption{Properties of the solutions of the wave equation $\square_g \phi = 0$ in FLRW models with flat space sections.} \label{table:flat}
\end{center}
\end{table}

We now consider the family of flat FLRW metrics whose scale factor satisfies $a(t)=t^p$, with $p \ge 0$. These correspond to universes filled with a perfect fluid with linear equation of state $p_m = w \rho_m$ (and zero cosmological constant), where
\[
p = \frac{2}{3(1+w)}
\]
(see appendix~\ref{appendixA}). We summarize our (proved) conjecture for the decay of solutions of the wave equations in these spacetimes in figure~\ref{fig:decayflat}. The precise statement can be found in section~\ref{section6}.

\begin{figure}[H]
\centering
\includegraphics[width=0.85\textwidth]{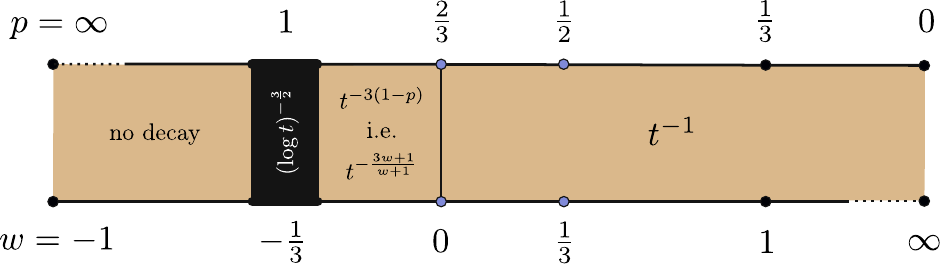}
\caption{Conjectured (and proved) decay rates in the flat case. The rates for $w \ge -\frac13$ follow from \cite{Nishihara, Wirth}, and the absence of decay for $w < -\frac13$ follows from \cite{CostaNatarioOliveira}. The Klainerman-Sarnak method gives the precise decay for the dust-filled ($w=0$) and radiation-filled ($w=\frac13$) universes.}
\label{fig:decayflat}
\end{figure}

\textbf{Hyperbolic case:} The results for FLRW universes with hyperbolic spatial sections are summarized in table~\ref{table:hyperbolic}. The decay rate in the dust-filled universe corrects the claim in \cite{AbbasiCraig}.

\renewcommand*{\arraystretch}{1.5}
\begin{table}[H]
\footnotesize
\begin{center}
\begin{tabulary}{\textwidth}{c|c|c|c|c} 
\toprule  
 Spacetime & Scale factor & $t(\tau)$ & Large times & Huygens \\[.4em]
 \hline 
 \rowcolor{gray!10}
 Dust & $a(\tau) = \cosh\tau-1$ & $t = \sinh\tau - \tau$ & $|\phi| \lesssim t^{-\frac{3}{2}}$ & No \\  
 Radiation & $a(\tau) = \sinh \tau$ & $t = \cosh\tau$ & $|\phi| \lesssim t^{-2}$ & Yes (strong) \\
 \rowcolor{gray!10}
 AdS & $a(\tau) = \sech\tau$ & $t = 2\arctan\left(\tanh\left( \frac{\tau}{2} \right) \right)$ & $\phi \to c$ & Yes (incomplete)\\
 de Sitter & $a(\tau) = -\cosech \tau$ & $t =-\log\left(-\tanh\left( \frac{\tau}{2} \right) \right)$ & $|\partial_t \phi| \lesssim e^{-2t}$ & Yes (incomplete) \\
 \rowcolor{gray!10}
 Milne & $a(\tau) = e^{\tau}$ & $t=e^{\tau}$ & $|\phi| \lesssim t^{-2}$ & Yes (strong) \\
\bottomrule
\end{tabulary}
\caption{Properties of the solutions of the wave equation $\square_g \phi = 0$ in FLRW models with hyperbolic space sections.}
\label{table:hyperbolic}
\end{center}
\end{table}

We now consider the family of hyperbolic FLRW universes filled with a perfect fluid with linear equation of state $p_m = w \rho_m$ (and zero cosmological constant -- see appendix~\ref{appendixA}). We summarize our conjecture for the decay of solutions of the wave equations in these spacetimes in figure~\ref{fig:decayhyperb}. The precise statement can be found in section~\ref{section6}.

\smallskip
\smallskip

\begin{figure}[H]
\centering
\includegraphics[width=0.9\textwidth]{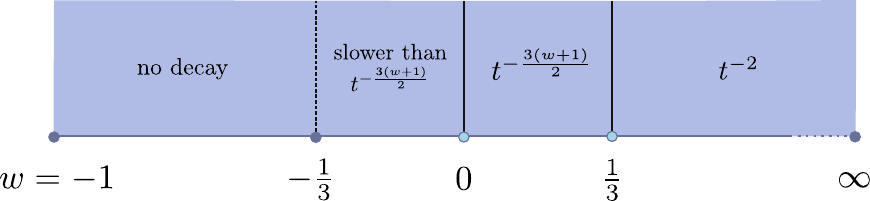}
\caption{Conjectured decay rates in the hyperbolic case. The Klainerman-Sarnak method gives the precise decay for the dust-filled ($w=0$) and radiation-filled ($w=\frac13$) universes.}
\label{fig:decayhyperb}
\end{figure}

\textbf{Spherical case:} The results for FLRW universes with spherical spatial sections are summarized in table~\ref{table:spherical}. The solutions in the dust-filled and radiation-filled universes blow up for {\em generic} initial data due to the Big Crunch singularity (see section~\ref{subsection1.4}).

\begin{table}[H]
\footnotesize
\begin{center}
\begin{tabulary}{\textwidth}{c|c|c|c|c} 
\toprule  
 Spacetime & Scale factor & $t(\tau)$ & Large times & Huygens \\[.4em]
 \hline 
 \rowcolor{gray!10}
 Dust & $a(\tau) = 1-\cos \tau$  & $t = \tau-\sin \tau$ & $|\phi|$ blows up*  & No \\  
 Radiation & $a(\tau) = \sin \tau$ & $t =1-\cos \tau$ & $|\phi|$ blows up*  & Yes (strong) \\
 \rowcolor{gray!10}
 de Sitter & $a(t) =\sec\tau$ & $t = 2 \arctanh\left( \tan \left(\frac{\tau}{2} \right) \right) $ & $|\partial_t \phi| \lesssim e^{-2t}$ & Yes (incomplete) \\
\bottomrule
\end{tabulary}
*generically (it is possible to choose initial data such that $|\phi|$ does not blow up)
\caption{Properties of the solutions of the wave equation $\square_g \phi = 0$ in FLRW models with spherical space sections.}
\label{table:spherical}
\end{center}
\end{table}

\subsection{A note on regularity}
In this paper we will assume that the initial data for the wave equation is smooth, so that the solutions are also smooth. Moreover, we will assume that the initial data belongs to any required Sobolev space, or even that it has compact support. All these assumptions can be easily relaxed, but we shall refrain from doing so as this would distract us from the main objectives of this paper.
%
%
\section{Condition for the existence of operators}\label{section1}
We are interested in solving the wave equation,
\begin{equation}
\square_g \phi = 0 \, ,
\end{equation}
in a $(3+1)$-dimensional FLRW background $(\bbR \times \Sigma_3, g)$, with metric 
\begin{equation} \label{flrwmetric3}
g = -dt^2 + a^2(t) d\varSigma_3^{\,2} \, ,
\end{equation}
where $d\varSigma_3^{\,2}$ is the standard Riemannian metric for $\Sigma_3=\bbR^3, \mathbb{H}^3$ or $S^3$. Introducing the conformal time
\begin{equation}
\tau = \int \frac{dt}{a(t)} \, ,
\end{equation}
we can write the wave equation in the form
\begin{equation} \label{wavephi}
\partial_{\tau}^2 \phi + 2\frac{a'}{a} \partial_{\tau} \phi = \Delta_{\Sigma_3} \phi \, ,
\end{equation}
where $\Delta_{\Sigma_3}$ is the Laplacian operator for the metric $d\varSigma_3^{\,2}$ and
\begin{equation}
a'(\tau) = \frac{d}{d\tau} a(t(\tau)).
\end{equation}
Now, let $\hat O$ be an operator acting on the space of smooth functions that commutes with $\Delta_{\Sigma_3}$. If the assumption
\begin{equation} \label{operatorassumption}
\hat O \left ( \partial_{\tau}^2 + 2 \frac{a'}{a} \partial_{\tau} \right) \phi = \partial_{\tau}^2 \hat O \phi + K \hat O \phi
\end{equation}
holds, where $K \in \{-1, 0, 1\}$ represents the curvature of the spatial sections we are considering,
then \eqref{wavephi} implies that $\hat O \phi$ satisfies the conformally invariant wave equation:
\begin{equation} \label{waveO}
\partial_{\tau}^2 \hat O \phi = (\Delta_{\Sigma_3} - K)\hat O \phi.
\end{equation}
We will search a suitable operator $\hat O$ by taking
\begin{equation}
\hat O \phi = f(\tau) \partial_{\tau} \phi + g(\tau) \phi,
\end{equation}
where $f$ and $g$ are functions to be determined. Plugging this expression into \eqref{operatorassumption}, we obtain the condition
\begin{align} \label{mustbe0}
\left(2 \frac{a^{\prime}}{a} f-2 f^{\prime}\right) \partial_{\tau}^{2} \phi +  
\left[2 \frac{a^{\prime}}{a} g-2 g^{\prime}-f^{\prime \prime}+2\left(\frac{a^{\prime}}{a}\right)^{\prime} f - Kf \right] \partial_{\tau} \phi & \nonumber \\
-(g^{\prime \prime} + Kg) \phi & = 0.
\end{align}
We will now try to satisfy this condition for each of the three different types of spatial sections.
%
%
\section{Flat case}\label{section2}
Let $K=0$. Equation \eqref{mustbe0} holds if and only if
\begin{equation} \label{addconstraint}
\begin{cases}
f(\tau) = \kappa a(\tau), \\
g(\tau) = \alpha \tau + \beta, \\
\displaystyle 2\frac{a'}{a}g - 2g' - f'' + 2\left(\frac{a'}{a}\right)' f = 0,
\end{cases}
\end{equation}
with $\alpha, \beta, \kappa \in \bbR$.
Let us consider the case
\begin{equation}
a(\tau) = \tau^j
\end{equation}
with $\tau>0$. It can be easily checked that there are non-trivial solutions $f$ and $g$ of the constraint equations \eqref{addconstraint} only for $j \in \{- 1, 0, 1, 2 \}$, provided that $\alpha, \beta, \kappa$ are chosen appropriately. We will analyze each of these four possibilities next.
%
%
\subsection{Minkowski spacetime}
\label{subsection2.1}
If $j=0$, the constraint equations \eqref{addconstraint} force us to choose $\alpha=0$, so that the functions $f$ and $g$ become
\begin{equation}
\begin{cases}
f(\tau) = \kappa, \\
g(\tau) = \beta.
\end{cases} 
\end{equation}
In this case $a(\tau) \equiv 1$, and so our FLRW universe is simply Minkowski's spacetime. Indeed, if $\phi$ is a solution of the wave equation in Minkowski's spacetime, then 
\begin{equation}
\hat O \phi = \kappa \partial_{\tau} \phi + \beta \phi
\end{equation}
is a solution as well.
%
%
\subsection{Dust-filled flat universe}
\label{subsection2.3}
If $j = 2$, conditions \eqref{addconstraint} force the functions $f$ and $g$ to be given by
\begin{equation}
\begin{cases}
f(\tau) = \kappa \tau^2, \\
g(\tau) = 3 \kappa \tau,
\end{cases}
\end{equation}
and therefore
\begin{equation}
\hat O \phi = \tau^2 \partial_{\tau} \phi + 3\tau \phi.
\end{equation}
In this case $a(\tau) = \tau^{2}$ and
\begin{equation}
\frac{dt}{d\tau} = a(\tau) \Leftrightarrow t \propto \tau^{3},
\end{equation}
so that
\begin{equation} \label{dustflatscalefactor}
a(t) \propto t^{\frac{2}{3}}.
\end{equation}
This scale factor $a(t)$ is a solution of the Friedmann equations when $K=0$, $\Lambda=0$ and $w=0$ (see appendix~\ref{appendixA}), corresponding to a dust-filled flat FLRW model (also known as the Einstein-de Sitter universe).

Here and in what follows we will denote the initial data for the wave equation at $t=t_0$ by
\begin{equation}
\begin{cases}
\phi(t_0, x) \eqqcolon \phi_0(x), \\
\partial_t \phi(t_0, x) \eqqcolon \phi_1(x).
\end{cases}
\end{equation}
Therefore, the initial data for
\begin{equation} \label{operator_dust}
\hat O \phi = \tau^2 \partial_{\tau} \phi + 3 \tau \phi = \tau^4 \partial_t \phi + 3 \tau \phi
\end{equation}
at the initial conformal time $\tau = \tau_0$ is:
\begin{equation} \label{initialdataop_dust}
\begin{cases}
\hat O \phi (\tau_0, x) = \tau_0^{\,4}\, \phi_1(x) + 3 \tau_0 \phi_0(x) \eqqcolon \psi_0(x), \\
\partial_{\tau} \hat O \phi(\tau_0, x) = \tau_0^{\,2}\, \Delta \phi_0(x) + \tau_0^{\,3}\, \phi_1(x) + 3 \phi_0(x) \eqqcolon \psi_1(x).
\end{cases}
\end{equation}
Here we used equation \eqref{wavephi} to write
\begin{equation}\partial_{\tau} \hat O \phi = \partial_{\tau} (\tau^2 \partial_{\tau}\phi + 3 \tau \phi) = \tau^2 \Delta \phi + \tau^3 \partial_t \phi + 3\phi. \nonumber
\end{equation}
Using Kirchhoff's formula for $\hat O \phi$, we have
\begin{align} \label{kirchoffflat_master}
\hat O \phi(\tau, x) &= \frac{1}{4 \pi(\tau- \tau_0)^2} \int_{\partial B_{\tau-\tau_0}(x)} \psi_0(y)dV_2(y) \nonumber \\
&+\frac{1}{4 \pi (\tau - \tau_0)} \int_{\partial B_{\tau-\tau_0} (x)} \nabla \psi_0(y) \cdot \frac{y-x}{|y-x|} dV_2(y) \\
&+ \frac{1}{4\pi(\tau-\tau_0)} \int_{\partial B_{\tau-\tau_0} (x)} \psi_1(y) dV_2(y). \nonumber
\end{align}
An explicit expression for $\phi$ can be obtained by noticing that $\hat O \phi = \tau^{-1} \partial_{\tau}(\tau^3 \phi)$ and then integrating Kirchhoff's formula:
\begin{empheq}[box=\fbox]{align}
\phi(\tau, x)  &=\frac{\tau_0^{\,3}}{\tau^3} \phi_0(x)+  \frac{1}{\tau^3}\int_{\tau_0}^{\tau} \frac{s}{4\pi (s-\tau_0)^2} \int_{\partial B_{s - \tau_0} (x)} \psi_0(y) dV_2 (y) ds \nonumber \\
&+ \frac{1}{\tau^3} \int_{\tau_0}^{\tau} \frac{s}{4 \pi (s -\tau_0) } \int_{\partial B_{s -\tau_0} (x)} \nabla \psi_0(y) \cdot \frac{y-x}{|y-x|} dV_2(y) ds \label{kirchhoff2} \\
&+ \frac{1}{\tau^3} \int_{\tau_0}^{\tau} \frac{s}{4 \pi (s-\tau_0)} \int_{\partial B_{s- \tau_0} (x)} \psi_1(y) dV_2(y) ds. \nonumber
\end{empheq}
This is the formula obtained by Klainerman and Sarnak in \cite{KlainermanSarnak}. Abbasi and Craig used this formula in \cite{AbbasiCraig} to obtain decay estimates by noticing that
\begin{equation}
\left | \int_{\tau_0}^{\tau} \frac{s}{4 \pi (s-\tau_0)^2} \int_{\partial B_{s- \tau_0}(x)} \psi_0(y) dV_2(y) ds \right | \lesssim \| \psi_0 \|_{\infty} + \| \psi_0 \|_1
\end{equation}
(and similarly for the other integrals in \eqref{kirchhoff2}), since 
\begin{align}
& \left | \int_{\tau_0}^{2 \tau_0} \frac{s}{4\pi (s - \tau_0)^2} \int_{\partial B_{s-\tau_0} (x)} \psi_0(y) dV_2(y) ds  \right| \nonumber \\
& = \left | \int_{\tau_0} ^{2 \tau_0} \frac{s}{4\pi} \int_{S^2} \psi_0(x + (s-\tau_0)z) dV_2(z) ds \right | \lesssim \| \psi_0 \|_{\infty}
\end{align}
and 
\begin{equation}
\frac{s}{4 \pi (s-\tau_0)^2} \lesssim 1 \quad \text{for} \quad s > 2 \tau_0.
\end{equation}
Therefore
\begin{equation}
|\phi| \lesssim \tau^{-3} (\| \phi_0 \|_{\infty} + \|\psi_0 \|_{\infty} + \| \psi_0 \|_{1} + \| \nabla \psi_0 \|_{\infty} + \| \nabla \psi_0 \|_{1} + \| \psi_1 \|_{\infty} + \| \psi_1 \|_{1}),
\end{equation}
and so, provided that the relevant norms are finite, we have
\begin{equation}\label{flatdecaydust}
|\phi| \lesssim t^{-1}.
\end{equation}

From \eqref{kirchhoff2} we expect the solutions of the wave equation to diverge as $t \to 0$, which can be traced back to the fact that the FLRW metric is singular at $t = 0$. However, as noticed in \cite{AbbasiCraig} (see also \cite{GiraoNatarioSilva, Ringstrom1}), it is possible to construct special solutions which have a well-defined limit as $t \to 0$. Indeed, if we fix $\phi_0$, $\phi_1$ and $\tau$, and take the limit $\tau_0 \to 0$ in equations \eqref{initialdataop_dust} and \eqref{kirchhoff2}, we obtain
\begin{equation}
\phi(\tau, x) = \frac{3}{4\pi\tau^3} \int_{B_{\tau}(x)} \phi_0(y) dV_3(y).
\end{equation}
It is clear that this expression determines a solution of the wave equation satisfying
\begin{equation}
\lim_{\tau \to 0} \phi(\tau,x) = \phi_0(x).
\end{equation}
%
%
\subsection{Radiation-filled flat universe} \label{subsection2.4}
If $j = 1$, conditions \eqref{addconstraint} tell us that the functions $f$ and $g$ must be of the form
\begin{equation} \label{fgradiation}
\begin{cases}
f(\tau) = \kappa \tau, \\
g(\tau) = \alpha \tau + \kappa.
\end {cases}
\end{equation}
In this case $a(\tau) = \tau$ and
\begin{equation}
\frac{dt}{d\tau} = a(\tau) \Leftrightarrow t \propto \tau^{2},
\end{equation}
which implies
\begin{equation} \label{aradiation}
a(\tau) \propto t^{\frac{1}{2}}.
\end{equation}
This scale factor is a solution of the Friedmann equations for $K=0$, $\Lambda=0$ and $w=\frac{1}{3}$ (see appendix~\ref{appendixA}). Therefore, this model corresponds to a flat, radiation-filled universe.
We notice that in this case $\hat O$ can be taken as the multiplicative operator,
\begin{equation}
\hat O \phi = \tau \phi.
\end{equation}
The reason for this is that the energy-momentum tensor of a radiation fluid is traceless, and thus the scalar curvature $R$ of the $4$-dimensional spacetime vanishes.
Consequently, the wave equation coincides with the conformally invariant wave equation, and so the scalar field
$
a \phi = \tau \phi
$
satisfies the wave equation in Minkowski spacetime. 
 
 The initial data for $\hat O \phi = \tau \phi$ at the initial conformal time $\tau = \tau_0$ is given by
 \begin{equation}\label{operatorinitdata_radiation}
 \begin{cases}
\hat O \phi(\tau_0, x) = \tau_0 \phi_0(x) \eqqcolon  \psi_0(x), \\
\partial_{\tau} \hat O \phi(\tau_0, x) = \tau_0^{\,2} \phi_1(x) + \phi_0(x) \eqqcolon  \psi_1(x).
 \end{cases}
 \end{equation}
 Therefore, if the appropriate Sobolev norms of the initial data are finite, by the standard decay of the wave equation in Minkowski spacetime (see e.g. \cite{Hormander}) we have, as $\tau \to +\infty$,
 \begin{equation}\label{flatdecayradiation}
 |\hat O \phi| \lesssim \tau^{-1} \Leftrightarrow |\phi| \lesssim \tau^{-2} \sim t^{-1},
 \end{equation}
 the same decay rate as in Minkowski spacetime.
 
 The explicit expression for $\phi$ is again given by Kirchhoff's formula \eqref{kirchoffflat_master}:
\begin{empheq}[box=\fbox]{align}
\phi(\tau, x) &= \frac{1}{4 \pi \tau (\tau- \tau_0)^2} \int_{\partial B_{\tau-\tau_0}(x)} \psi_0(y)dV_2(y) \nonumber \\
&+ \frac{1}{4 \pi \tau(\tau - \tau_0)} \int_{\partial B_{\tau-\tau_0} (x)} \nabla \psi_0(y) \cdot \frac{y-x}{|y-x|} dV_2(y) \label{kirchhoff3} \\
&+ \frac{1}{4\pi\tau(\tau-\tau_0)} \int_{\partial B_{\tau-\tau_0} (x)} \psi_1(y) dV_2(y). \nonumber
\end{empheq}
As one would expect, $\phi$ satisfies the strong Huygens principle; a related result was obtained in \cite{KulczyckiMalec, MalecWylezek} for gravitational waves in the Regge-Wheeler gauge.

Again it is possible to construct solutions with a well-defined limit as $t \to 0$ by fixing $\phi_0$, $\phi_1$ and $\tau$ and taking the limit $\tau_0 \to 0$ in equations \eqref{operatorinitdata_radiation} and \eqref{kirchhoff3}:
\begin{equation} 
\phi(\tau, x) = \frac{1}{4\pi\tau^2} \int_{\partial B_{\tau}(x)} \phi_0(y) dV_2(y).
\end{equation}
It is easy to check that this expression determines a solution of the wave equation satisfying
\begin{equation}
\lim_{\tau \to 0} \phi(\tau,x) = \phi_0(x).
\end{equation}
%
%
\subsection{Flat de Sitter universe} \label{subsection2.2}
If $j = - 1$, we can choose $\alpha=\beta=0$, so that the functions $f$ and $g$ become
\begin{equation}
\begin{cases}
f(\tau) = \kappa \tau^{-1}, \\
g(\tau) = 0.
\end{cases}
\end{equation}
These lead to the operator $\hat O  = \tau^{-1} \partial_{\tau} $.
In this case $a(\tau) = \tau^{- 1}$, and so
\begin{equation}
\frac{dt}{d\tau} = a(\tau) \Leftrightarrow 
t = \log \tau.
\end{equation}
This yields
\begin{equation} \label{adesitter}
a(t) = e^{-t},
\end{equation}
corresponding to the contracting flat de Sitter spacetime (which is a globally hyperbolic region of the full de Sitter universe); the expanding case can be trivially recovered by changing the time coordinate from $t$ to $-t$. 

The initial data for 
\begin{equation}
\hat O \phi = \tau^{-1} \partial_{\tau} \phi  = \tau^{-2} \partial_t \phi
\end{equation}
at the corresponding conformal time $\tau = \tau_0$ is given by
\begin{equation}
\begin{cases}
\hat O \phi(\tau_0, x) = \tau_0^{-2} \phi_1 (x) \eqqcolon \psi_0(x), \\
\partial_{\tau} \hat O \phi(\tau_0, x) = \tau_0^{-1} \Delta \phi_0(x) + \tau_0^{-3} \phi_1(x) \eqqcolon \psi_1 (x),
\end{cases}
\end{equation}
where we used \eqref{wavephi} to write
\begin{equation}
\partial_{\tau} \hat O \phi = \partial_{\tau}(\tau^{-1} \partial_{\tau} \phi) 
=\tau^{-1} \Delta \phi + \tau^{-3} \partial_t \phi.
\end{equation}
If the appropriate Sobolev norms of the initial data are finite then we have, by boundedness of the solution of the wave equation in the region $0 < \tau \leq \tau_0$ of Minkowski's spacetime,
\begin{equation}
|\hat O \phi|  \lesssim 1 \Leftrightarrow |\partial_t \phi| \lesssim \tau^2 = e^{2t}
\end{equation}
as $t \to - \infty$ (that is, as $\tau \to 0$). This translates to $|\partial_t \phi| \lesssim e^{-2t}$ as $t \to +\infty$ in the expanding case, confirming the result in \cite{NatarioSasane}.

An explicit expression for $\phi$ can be obtained by integrating Kirchhoff's formula \eqref{kirchoffflat_master}:
\begin{empheq}[box=\fbox]{align}
\phi(\tau, x) &= \phi_0(x) + \int_{\tau_0}^{\tau} \frac{s}{4 \pi(s- \tau_0)^2} \int_{\partial B_{s-\tau_0}(x)} \psi_0(y)dV_2(y) ds \nonumber \\
&+\int_{\tau_0}^{\tau}\frac{s}{4 \pi (s - \tau_0)} \int_{\partial B_{s-\tau_0} (x)} \nabla \psi_0(y) \cdot \frac{y-x}{|y-x|} dV_2(y) ds \\
&+ \int_{\tau_0}^{\tau}\frac{s}{4\pi(s-\tau_0)} \int_{\partial B_{s-\tau_0} (x)} \psi_1(y) dV_2(y) ds. \nonumber
\end{empheq}
In the particular case when $\phi_1 \equiv 0$, we obtain
\begin{align}
\phi(\tau, x) & = \phi_0(x) + \int_{\tau_0}^{\tau}\frac{s}{4\pi\tau_0(s-\tau_0)} \int_{\partial B_{s-\tau_0} (x)} \Delta\phi_0(y) dV_2(y) ds \nonumber \\
& = \phi_0(x) + \frac{1}{4\pi\tau_0} \int_{B_{\tau-\tau_0} (x)} \left( 1 + \frac{\tau_0}{|z-x|}\right) \Delta\phi_0(z) dV_3(z).
\end{align}
Using the divergence theorem and Green's identities, this expression can be rewritten as
\begin{align}
\phi(\tau, x) & = \frac{\tau}{4 \pi \tau_0 (\tau - \tau_0)} \int_{\partial B_{\tau-\tau_0} (x)} \nabla \phi_0(y) \cdot \frac{y-x}{|y-x|} dV_2(y) \nonumber \\
& + \frac{1}{4 \pi (\tau - \tau_0)^2} \int_{\partial B_{\tau-\tau_0} (x)} \phi_0(y) dV_2(y) .
\end{align}
In particular, $\phi(\tau, x)$ depends only on the values of $\phi_0$ and $\nabla \phi_0$ on $\partial B_{\tau-\tau_0} (x)$, that is, $\phi$ satisfies the strong Huygens principle in the special case when $\phi_1 \equiv 0$. This fact was previously noticed by Yagdjian in \cite{Yagdjian}, who dubbed this property the {\em incomplete Huygens principle}.
%
%
\section{Hyperbolic case}
\label{section3}
Next, we want to find suitable operators $\hat O$ in the case $K=-1$. We start by noticing that the constraints \eqref{mustbe0} require the functions $f$ and $g$ in operators of the form $\hat O = f(\tau) \partial_{\tau} + g(\tau)$ to satisfy
\begin{equation}\label{def_fg_hyp}
\begin{cases}
f(\tau) = \kappa a(\tau), \\
g(\tau) = \alpha \cosh(\tau) + \beta \sinh(\tau), \\
\displaystyle 2 \frac{a^{\prime}}{a} g-2 g^{\prime}-f^{\prime \prime}+2\left(\frac{a^{\prime}}{a}\right)^{\prime} f + f = 0.
\end{cases}
\end{equation}
These conditions can be met for certain choices of scale factors $a(\tau)$ satisfying the Friedmann equations (see appendix~\ref{appendixA}). In these cases, it is possible to obtain an explicit formula for $\phi$ from the expression of $\hat O \phi$ in terms of spherical means, analogous to Kirchhoff's formula in $\bbR^3$, which we now discuss. 
%
%
\subsection{Spherical means formula for $K=-1$}
\label{subsection3.1}
We will now obtain a spherical means formula for the solution of the conformally invariant wave equation \eqref{waveO} in the spacetime given by the metric
\begin{equation}
g = -dt^2+d \varSigma_3^{\, 2},
\end{equation}
where $d\varSigma_3^{\, 2}$ is the line element for the hyperbolic space $\mathbb{H}^3$. What follows is a simplified version of the more general derivation in \cite{KlainermanSarnak}. 

Consider the Cauchy problem
\begin{equation} \label{wavehyperb}
\begin{cases}
\partial_t^2 \phi = \Delta\phi+\phi, \\
\phi(0, x) = \phi_0(x), \\
\partial_t \phi(0, x) = \phi_1(x),
\end{cases}
\end{equation}
where $\Delta$ is the Laplacian operator in $\mathbb{H}^3$. Using geodesic polar coordinates about some point $x \in \mathbb{H}^3$, we know that
\begin{equation} \label{spatiallineel}
d\varSigma_3^{\,2} = dr^2 + \sinh^2(r)d\varOmega^2,
\end{equation}
where $d\varOmega^2$ is the line element for the unit sphere $S^2$, and also that
\begin{equation} \label{laplacianhyperb}
\Delta = \pdv[2]{r}+ 2 \cotanh(r) \pdv{r} + \frac{1}{\sinh^2(r)} \Delta_{S^2}, 
\end{equation}
where $\Delta_{S^2}$ is the Laplacian on $S^2$.
The geodesic sphere $S(x, r)$ about the point $x$ is defined as
\begin{equation}\label{geodesicdef}
S(x, r) \coloneqq \exp_x \left( \{z \in T_x\mathbb{H}^3 : \|z\|= r \} \right) \equiv \exp_x(\partial B_r(0)),
\end{equation}
where $\exp_x\colon T_x\mathbb{H}^3 \to \mathbb{H}^3$ is the geodesic exponential map and $B_r(0)$ is the ball of radius $r$ about $0$ in $T_x\mathbb{H}^3$.
We define the {\em spherical mean} of a function $\phi(t, x)$ on the geodesic sphere of radius $r$ about $x$ as 
\begin{align} \label{sphericalmean}
M_{\phi}(t, r, x) &\coloneqq \frac{1}{4 \pi \sinh^2(r)} \int_{S(x, r)} \phi(t, y) dV_2(y) \nonumber \\
&= \frac{1}{4\pi}\int_{S^2} \phi(t, \exp_x(rz)) dV_2(z),
\end{align}
where $S^2$ is the unit sphere in $T_x\mathbb{H}^3$. Note that this last formula allows us to extend the spherical mean to negative values of the second argument, yielding
\begin{equation}
M_{\phi}(t, -r, x)=M_{\phi}(t, r, x).
\end{equation}
Now, if $\phi$ is a solution of the problem \eqref{wavehyperb}, we have
\begin{equation}\label{m_wavehyperb}
\pdv[2]{t} M_{\phi}(t, r, x) 
 = \left( \pdv[2]{r}+2\cotanh(r) \pdv{r} + 1 \right) M_{\phi}(t, r, x),
\end{equation}
where we used \eqref{sphericalmean}, \eqref{wavehyperb} and the divergence theorem on the sphere.
From \eqref{wavehyperb}, we have the initial data 
\begin{equation}
\begin{cases}
M_{\phi}(0, r, x)=M_{\phi_0}(r; x), \\
\partial_t M_{\phi}(0, r, x) = M_{\phi_1}(r; x).
\end{cases}
\end{equation} 
Using \eqref{m_wavehyperb}, it can be checked that $\omega(t,r,x) \coloneqq \sinh(r) M_{\phi}(t,r,x)$ satisfies
\begin{equation} \label{telegraph}
\partial_t^2 \omega=\partial_r^2 \omega.
\end{equation}
Given the initial data 
\begin{equation}
\begin{cases}
\omega(0,r,x) = \sinh(r)M_{\phi_0}(r,x) \eqqcolon \gamma(r,x), \\
\partial_t \omega(0,r,x) = \sinh(r) M_{\phi_1}(r,x) \eqqcolon \psi(r,x),
\end{cases}
\end{equation}
the solution of equation \eqref{telegraph} is given by the d'Alembert formula:
\begin{equation}
\omega(t,r,x) = \frac{1}{2}\left[\gamma(r+t,x)+\gamma(r-t,x)+\int_{r-t}^{r+t} \psi(s,x)ds  \right].
\end{equation}
Since $M_\phi(t, r, x) \to \phi(t, x)$ as $r \to 0$, we obtain
\begin{align}
\phi(t, x) &= \lim_{r \to 0} \frac{1}{2 \sinh(r)} \left[\gamma(r+t,x)+\gamma(r-t,x)+\int_{r-t}^{r+t} \psi(s,x)ds  \right] \nonumber \\
&= \lim_{r \to 0} \frac{1}{2 \cosh(r)} \left[\gamma'(r+t,x) + \gamma'(r-t,x) + \psi(r+t,x)-\psi(r-t,x) \right] \nonumber \\
&=\gamma'(t,x)+\psi(t,x) = \partial_t\left[\sinh(t)M_{\phi_0}(t,x)\right] + \sinh(t)M_{\phi_1}(t,x),
\end{align}
where we used L'H\^{o}pital's rule and the fact that $\gamma$ and $\psi$ are odd functions. From \eqref{sphericalmean}, the solution of the Cauchy problem \eqref{wavehyperb} is then given by
\begin{empheq}[box=\fbox]{align} \label{hyperbolicsphericalmean2}
\phi(t, x) &= \frac{\cosh(t)}{4 \pi} \int_{S^2} \phi_0 \big(\exp_x(tz) \big)\, dV_2(z) \nonumber \\
& + \frac{\sinh(t)}{4 \pi} \int_{S^2} d\phi_0(\dot c_z(t)) \, dV_2(z) \\
& + \frac{\sinh(t)}{4 \pi} \int_{S^2} \phi_1 \big(\exp_x(tz) \big) \, dV_2(z), \nonumber
\end{empheq}
with $c_z(t)=\exp_x(tz)$.
%
%
\subsection{Dust-filled hyperbolic universe}\label{subsection3.2}
Let us consider the case
\begin{equation}
a(\tau ) = \cosh(\tau)-1
\end{equation}
with $\tau>0$. This corresponds to the dust-filled hyperbolic model, since $a(\tau)$ solves Friedmann's equations for $K=-1$, $\Lambda=0$ and $w = 0$ (see appendix~\ref{appendixA}).
With this choice of the scale factor, the wave equation \eqref{wavephi} becomes
\begin{equation} \label{wavehyperbdust}
\partial_{\tau}^2 \phi + \frac{2 \sinh(\tau)}{\cosh(\tau)-1} \partial_{\tau} \phi = \Delta \phi.
\end{equation}
In this case, the constraint equations \eqref{def_fg_hyp} are satisfied if we choose $\kappa = 1$, $\alpha = 0$ and $\beta = \frac{3}{2}$, that is, if we choose
\begin{equation}\label{hyperbdustoperatordef}
\hat O  \phi = (\cosh(\tau) - 1) \partial_{\tau} \phi + \frac{3}{2} \sinh(\tau) \phi.
\end{equation}
We notice that we can rewrite this as
\begin{equation} \label{hyperbdustoperator}
\hat O \phi = \frac{2}{\sinh \left (\frac{\tau}{2} \right)} \partial_{\tau}\left(\sinh^3 \left(\frac{\tau}{2} \right ) \phi \right),
\end{equation}
which is the expression that can be found in \cite{KlainermanSarnak}.
The initial data for $\hat O \phi$ at the initial conformal time $\tau = \tau_0$ is then
\begin{equation}\label{initialdatahyperbdust}
\begin{cases}
\hat O \phi(\tau_0, x)=(\cosh(\tau_0)-1) \phi_1(x) +  \frac{3}{2} \sinh(\tau_0) \phi_0(x) \eqqcolon \psi_0(x), \\[0.5em]
\partial_{\tau} \hat O \phi(\tau_0, x) = \frac{1}{2} \sinh(\tau_0) (\cosh(\tau_0) - 1) \phi_1(x) \\ 
\qquad \qquad \qquad \quad + (\cosh(\tau_0)-1)\Delta \phi_0(x) + \frac{3}{2} \cosh(\tau_0) \phi_0(x) \eqqcolon \psi_1(x), 
\end{cases}
\end{equation}
where we used \eqref{wavehyperbdust} to write
\begin{equation}
\partial_{\tau} \hat O \phi = \frac{1}{2} \sinh(\tau) (\cosh(\tau) - 1) \partial_t \phi + (\cosh \tau-1) \Delta \phi + \frac{3}{2}\cosh(\tau) \phi.
\end{equation}
Then \eqref{hyperbolicsphericalmean2} and \eqref{hyperbdustoperator} give the following expression for the solution:
\begin{empheq}[box=\fbox]{align}
\label{hyperbdustsolution}
\sinh^3&\left(\frac{\tau}{2} \right) \phi(\tau, x) = \sinh^3 \left(\frac{\tau_0}{2} \right) \phi_0(x) \nonumber \\
&+ \int_{\tau_0}^{\tau} \frac{\cosh(s-\tau_0)}{8\pi} \sinh \left( \frac{s}{2} \right)\int_{S^2} \psi_0 \big (\exp_x ((s-\tau_0)z) \big) \, dV_2(z) ds \nonumber \\
&+ \int_{\tau_0}^{\tau}\frac{\sinh(s-\tau_0)}{8\pi}  \sinh \left( \frac{s}{2} \right) \int_{S^2} d\psi_0(c'_z(s-\tau_0)) \, dV_2(z) ds \\
&+ \int_{\tau_0}^{\tau}\frac{\sinh(s-\tau_0)}{8\pi}  \sinh \left( \frac{s}{2} \right ) \int_{S^2} \psi_1 \big (\exp_x ((s-\tau_0)z) \big) \, dV_2(z) ds,  \nonumber
\end{empheq}

Let us analyze the decay of the solution for $\tau \to +\infty$. The first integral, for instance, can be estimated by noting that
\begin{equation}
\left | \int_{\tau_0}^{2\tau_0}  \cosh(s-\tau_0) \sinh\left (\frac{s}{2} \right)\int_{S^2} \psi_0 \big (\exp_x ((s-\tau_0)z) \big)\, dV_2(z) ds  \right| \lesssim \| \psi_0 \|_{\infty}
\end{equation}
and that, rewriting the integral on $S^2$ as an integral over the geodesic sphere,
\begin{align}
& \left | \int_{2\tau_0}^{\tau}  \cosh(s-\tau_0) \sinh\left (\frac{s}{2} \right)\int_{S^2} \psi_0 \big (\exp_x ((s-\tau_0)z) \big)\, dV_2(z) ds  \right| \nonumber \\
& = \left | \int_{2\tau_0}^{\tau} \frac{\cosh(s-\tau_0)}{\sinh^2(s-\tau_0)} \sinh\left( \frac{s}{2} \right) \int_{S(x, s - \tau_0)}  \psi_0(y) dA(y) ds \right | \lesssim \| \psi_0 \|_{1}.
\end{align}
Similar computations apply to the other integrals in expression \eqref{hyperbdustsolution}. After dividing both sides by $\sinh^3 \left( \frac{\tau}{2} \right)$, we can write
\begin{equation}\label{decaytoimprove}
|\phi| \lesssim e^{-\frac{3}{2} \tau} (\| \phi_0 \|_{\infty} + \|\psi_0 \|_{\infty} + \| \psi_0 \|_{1} + \| d \psi_0 \|_{\infty} + \| d \psi_0 \|_{1} + \| \psi_1 \|_{\infty} + \| \psi_1 \|_{1}),
\end{equation}
and so, provided that the relevant norms are finite, we have
\begin{equation}\label{hypdecaydust}
|\phi| \lesssim e^{-\frac{3}{2} \tau}.
\end{equation}
We notice that, in general, a faster decay rate is not possible (see appendix~\ref{appendixB}). To obtain this estimate in terms of the coordinate time $t$, we see from
\begin{equation}
t=\int a(\tau) d\tau=  \sinh(\tau)-\tau
\end{equation}
that, as $\tau \to +\infty$,
\begin{equation}
t \sim e^{\tau},
\end{equation}
whence
\begin{equation}
|\phi| \lesssim t^{-\frac{3}{2}}.
\end{equation}

As in the flat case, there exist solutions which have a well-defined limit at the Big Bang: if we take the limit $\tau_0 \to 0$ while keeping $\phi_0$, $\phi_1$ and $\tau$ fixed, we obtain from \eqref{initialdatahyperbdust} and \eqref{hyperbdustsolution} the limit solution
\begin{equation}
\phi\left (\tau, x \right) = \frac{3}{16\pi\sinh^3\left( \frac{\tau}{2} \right)} \int_0^{\tau} \sinh(s) \sinh\left( \frac{s}{2} \right) \int_{S^2} \phi_0(\exp_x(zs)) dV_2(z) ds.
\end{equation}
Using L'H\^{o}pital's rule, one can easily check that
\begin{equation}
\lim_{\tau \to 0} \phi(\tau, x) =  \phi_0(x).
\end{equation}
%
%
\subsection{Radiation-filled hyperbolic universe}
\label{subsection3.3}
Let us now consider the case
\begin{equation}
a(\tau ) = \sinh(\tau)
\end{equation}
with $\tau>0$. This corresponds to the radiation-dominated hyperbolic universe, since $a(\tau)=\sinh(\tau)$ solves the Friedmann equations for $K=-1$, $\Lambda=0$ and $w = \frac{1}{3}$ (see appendix~\ref{appendixA}).
With this choice of the scale factor, our wave equation \eqref{wavephi} becomes
\begin{equation}\label{hyperbradiationwaveeqn}
\partial_{\tau}^2 \phi + 2 \cotanh(\tau) \partial_{\tau} \phi = \Delta \phi.
\end{equation}
In this case, the constraint equations \eqref{def_fg_hyp} are satisfied if we choose $\alpha = \kappa = 0$ and $\beta = 1$, that is, if we choose
\begin{equation}
\hat O \phi = \sinh(\tau) \phi.
\end{equation}
Again, this operator can be taken as multiplicative because in this universe the wave equation coincides with the conformally invariant wave equation.

The initial data for $\hat O \phi$ at the initial conformal time $\tau=\tau_0$ is
\begin{equation}\label{hyperbradiationinitdata}
\begin{cases}
\hat O \phi(\tau_0, x) = \sinh(\tau_0) \phi_0(x) \eqqcolon \psi_0(x), \\
\partial_{\tau} \hat O \phi(\tau_0, x) = \cosh(\tau_0) \phi_0(x) + \sinh^2(\tau_0) \phi_1(x) \eqqcolon \psi_1(x).
\end{cases}
\end{equation}
The solution in terms of spherical means is then given by \eqref{hyperbolicsphericalmean2}:
\begin{empheq}[box=\fbox]{align} \label{hyperbradiationsphericalmeans}
\sinh(\tau) \phi(\tau,x) &= \frac{\cosh(\tau-\tau_0)}{4\pi} \int_{S^2} \psi_0 \big(\exp_x((\tau-\tau_0)z) \big) \, dV_2(z) \nonumber \\ 
&  + \frac{\sinh(\tau-\tau_0)}{4\pi} \int_{S^2} d \psi_0(c'_z(\tau-\tau_0)) \, dV_2(z) \\
&  + \frac{\sinh(\tau-\tau_0)}{4\pi} \int_{S^2} \psi_1 \big(\exp_x((\tau-\tau_0)z) \big) \, dV_2(z). \nonumber
\end{empheq}
As in the flat case, $\phi$ satisfies the strong Huygens principle.

Let us now analyze the decay of the solution as $\tau \to +\infty$. We notice that, after dividing the formula above by $\sinh(\tau)$, the first term of the solution can be estimated as 
\begin{align}
|\phi| &= \left | \frac{\cosh(\tau-\tau_0)}{4\pi\sinh(\tau)\sinh^2(\tau-\tau_0)} \int_{S(x, \tau-\tau_0)}\psi_0(y) \, dA(y) \right |  \\
& \lesssim \left (\|\psi_0\|_{L^1(\mathbb{H}^3)} + \| d\psi_0 \|_{L^1(\mathbb{H}^3)} \right ) e^{-2\tau}. \nonumber
\end{align}
In fact, assuming for simplicity that the initial data has compact support, we have
\begin{align}
\int_{S(x, R)}&\psi_0(y) dA(y) = -\int_{R}^{+\infty} \dv{r} \left( \int_{S(x, r)} \psi_0(y) dA(y) \right) dr \nonumber \\
= &-\int_R^{+\infty} 2\sinh(r) \cosh(r) \left( \int_{S^2} \psi_0 \left(\exp_x(rz)\right) dV_2(z)  \right) dr  \\
&- \int_R^{+\infty} \sinh^2(r) \left( \int_{S^2} d\psi_0(\dot c_z(r)) dV_2(z) \right) dr, \nonumber
\end{align}
where $c_z(r) = \exp_x(rz)$, whence
\begin{equation}
\left|\int_{S(x, R)} \psi_0(y) dA(y)\right| \lesssim \|\psi_0\|_{L^1(\mathbb{H}^3)} + \| d\psi_0 \|_{L^1(\mathbb{H}^3)}
\end{equation}
for all $R>0$. The remaining two terms have the same decay rate, which can be found after similar computations.
Therefore, if the relevant Sobolev norms of the initial data are finite, we have
\begin{equation}
|\phi| \lesssim e^{-2\tau}.
\end{equation}
To obtain the decay rate in terms of the time variable $t$, we use
\begin{equation}
t = \int \sinh(\tau) d\tau = \cosh(\tau) \sim e^\tau,
\end{equation}
so that, as $t \to +\infty$,
\begin{equation}
|\phi| \lesssim t^{-2}.
\end{equation}

Again there exist solutions with a well-defined limit at the Big Bang: if we take the limit $\tau_0 \to 0$ while keeping $\phi_0$, $\phi_1$ and $\tau$ fixed, we obtain from \eqref{hyperbradiationinitdata} and \eqref{hyperbradiationsphericalmeans} the limit solution
\begin{equation}
\phi(\tau, x)=\frac{1}{4\pi} \int_{S^2}\phi_0 \big(\exp_x(\tau z) \big) dV_2(z),
\end{equation}
which clearly converges to  $\phi_0(x)$ as $\tau \to 0$.
%
%
\subsection{Anti-de Sitter universe}
\label{subsection3.4}
Let us now analyze the case
\begin{equation}
a(\tau ) = \sech(\tau).
\end{equation}
This corresponds to a globally hyperbolic region of the anti-de Sitter universe, since $a(\tau)= \sech(\tau)$ is a solution for the Friedmann equations for $K=-1$, $\Lambda=-3$ and $\rho_0= 0$  (see appendix~\ref{appendixA}). We can obtain the scale factor as a function of the physical time $t$ by noticing that
\begin{equation}
t= \int \frac{d\tau}{\cosh(\tau)} = 2 \arctan \left( \tanh \left ( \frac{\tau}{2} \right) \right),
\end{equation}
and so
\begin{equation}
a(t) = \sech \left( 2 \arctanh \left( \tan\left( \frac{t}{2} \right) \right) \right) = \cos (t),
\end{equation}
with $t \in \left (-\frac{\pi}{2}, \frac{\pi}{2} \right)$.  

With this choice of the scale factor, \eqref{wavephi} becomes
\begin{equation}\label{hyperbvacuumwaveeqn}
\partial_{\tau}^2 \phi - 2 \tanh(\tau) \partial_{\tau} \phi = \Delta \phi.
\end{equation}
We can satisfy the constraint equations \eqref{def_fg_hyp} by choosing $\alpha = \beta = 0$ and $\kappa = 1$, so that we have
\begin{equation}
\hat O \phi = \sech(\tau) \partial_{\tau} \phi.
\end{equation}
 As before, we can express the solution in terms of spherical means. We have the following initial data for $\hat O \phi$ at the initial conformal time $\tau = \tau_0$:
\begin{equation}\label{hyperbvacuuminitialdata1}
\begin{cases}
\hat O \phi(\tau_0, x) = \sech^2(\tau_0) \phi_1(x) \eqqcolon \psi_0(x), \\
\partial_{\tau} \hat O \phi(\tau_0, x) = \sech^2(\tau_0) \tanh(\tau_0) \phi_1(x) + \sech(\tau_0) \Delta \phi_0(x) \eqqcolon \psi_1(x),
\end{cases}
\end{equation}
where we used  \eqref{hyperbvacuumwaveeqn} to write
\begin{equation}
\partial_{\tau} \hat O \phi = \sech^2(\tau)\tanh(\tau)\partial_t \phi + \sech(\tau) \Delta \phi. 
\end{equation}
The spherical means solution is then
\begin{empheq}[box=\fbox]{align} \label{solutionAdS}
\phi(\tau, x) &=  \int_{\tau_0}^{\tau} \frac{\cosh(s-\tau_0)}{4\pi} \cosh(s) \int_{S^2} \psi_0 \big(\exp_x ((s-\tau_0)z) \big) \, dV_2(z) ds  \nonumber \\
&  + \int_{\tau_0}^{\tau} \frac{\sinh(s-\tau_0)}{4\pi} \cosh(s) \int_{S^2} d\psi_0(c'_z(s-\tau_0)) \, dV_2(z)  ds \\
&  + \int_{\tau_0}^{\tau} \frac{\sinh(s-\tau_0)}{4\pi} \cosh(s)\int_{S^2} \psi_1 \big (\exp_x ((s-\tau_0)z) \big ) \, dV_2(z) ds \nonumber \\
&  + \phi_0(x). \nonumber
\end{empheq}
In the particular case when $\phi_1 \equiv 0$ we obtain
\begin{align}
\phi(\tau, x) & = \phi_0(x) + \int_{0}^{\tau-\tau_0} \frac{\sinh(r)}{4\pi} \cosh(r+\tau_0)\int_{S^2} \sech(\tau_0) \Delta\phi_0 \big (\exp_x (rz) \big ) \, dV_2(z) dr \nonumber \\
& = \phi_0(x) + \frac{1}{4\pi} \int_{B_{\tau-\tau_0} (x)} \left( \tanh(\tau_0) + \cotanh(r) \right) \Delta\phi_0(y) dV_3(y),
\end{align}
where we defined $r(y)=\disth(y, x)$. Using the divergence theorem and Green's identities, together with $\Delta \cotanh(r) = 0$, this expression can be rewritten as
\begin{align}
\phi(\tau, x) & = \frac{\cosh(\tau)}{4 \pi \cosh(\tau_0) \sinh(\tau - \tau_0)} \int_{\partial B_{\tau-\tau_0} (x)} \nabla \phi_0(y) \cdot \nabla r(y) dV_2(y) \nonumber \\
& + \frac{1}{4 \pi \sinh^2(\tau - \tau_0)} \int_{\partial B_{\tau-\tau_0} (x)} \phi_0(y) dV_2(y).
\end{align}
Therefore, $\phi(\tau, x)$ depends only on the values of $\phi_0$ and $\nabla \phi_0$ on $\partial B_{\tau-\tau_0} (x)$, that is, $\phi$ satisfies Yagdjian's incomplete Huygens principle. In particular, if $\phi_0$ is compactly supported then we have
\begin{equation}
\lim_{\tau \to +\infty} \phi(\tau,x) = 0.
\end{equation}
Interestingly, in the case when $\phi_0 \equiv 0$ we have, for compactly supported $\phi_1$,
\begin{equation} \label{limitAdS}
\lim_{\tau \to +\infty} \phi(\tau,x) = - \frac{\sech^3(\tau_0)}{4\pi} \int_{\mathbb{H}^3} \phi_1(y) dV_3(y)
\end{equation}
(in particular this limit does not depend on $x$). By the superposition principle, this result holds in general for compactly supported initial data; it can be interpreted by noting that for fixed $x$ the curve $\tau \mapsto (\tau,x)$ approaches a single point in the full anti-de Sitter spacetime as $\tau \to +\infty$ (see \cite{Natario}).

To check \eqref{limitAdS}, we note that for $\phi_0 \equiv 0$ equation \eqref{solutionAdS} can be written in the form
\begin{align}
\phi(\tau, x) =  & - \frac{\partial}{\partial \tau_0} \int_{0}^{\tau-\tau_0} \frac{\sinh(r)}{4\pi} \cosh(r + \tau_0) \int_{S^2} \psi_0 \big(\exp_x (rz) \big) \, dV_2(z) dr  \nonumber \\
&  + \int_{0}^{\tau-\tau_0} \frac{\sinh(r)}{4\pi} \cosh(r+\tau_0)\int_{S^2} \psi_1 \big (\exp_x (r z) \big ) \, dV_2(z) dr.
\end{align}
Discarding the boundary term, whose limit will vanish for compactly supported initial data, and using \eqref{hyperbvacuuminitialdata1}, we obtain
\begin{equation}  
\lim_{\tau \to +\infty} \phi(\tau,x) = - \frac{\sech^3(\tau_0)}{4\pi} \int_{0}^{+\infty} \sinh^2(r) \int_{S^2} \phi_1 \big(\exp_x (rz) \big) \, dV_2(z) dr .
\end{equation}
%
%
\subsection{Hyperbolic de Sitter universe}
\label{subsection3.5}
Let us now consider the case
\begin{equation}
a(\tau ) = - \cosech(\tau)
\end{equation}
with $\tau < 0$. This corresponds to a globally hyperbolic region of the de Sitter universe, since $a(\tau)= -\cosech(\tau)$ is a solution for the Friedmann equations for $K=-1$, $\Lambda=-3$ and $\rho_0= 0$  (see appendix~\ref{appendixA}). We can obtain the scale factor as a function of the physical time $t$ by noticing that
\begin{equation}\label{confdesitterhyperb}
t = -\int \frac{d\tau}{\sinh(\tau)} = -\log \left( -\tanh \left ( \frac{\tau}{2} \right) \right),
\end{equation}
and so
\begin{equation}\label{scalefactordesitterhyperb}
a(t) = -\cosech(2 \arctanh(-e^{-t}))=\sinh(t),
\end{equation}
with $t \in (0, +\infty)$.

This choice of the scale factor in \eqref{wavephi} leads to
\begin{equation}\label{hyperbvacuumwaveeqn2}
\partial_{\tau}^2 \phi - 2 \cotanh(\tau) \partial_{\tau} \phi = \Delta \phi
\end{equation}
We can satisfy the constraint equations \eqref{def_fg_hyp} by choosing $\alpha = \beta = 0$ and $\kappa = -1$, so that
\begin{equation}
\hat O \phi = -\cosech(\tau) \partial_{\tau} \phi.
\end{equation}
Again we can repeat the same steps as in the previous cases to obtain the solution in terms of spherical means. The initial data for $\hat O \phi$ at the initial conformal time $\tau = \tau_0$ is
\begin{equation}\label{hyperbvacuuminitialdata2}
\begin{cases}
\hat O \phi(\tau_0, x) = \cosech^2(\tau_0) \phi_1(x) \eqqcolon \psi_0(x), \\
\partial_{\tau} \hat O \phi(\tau_0, x) = \cosech^2(\tau_0) \cotanh(\tau_0) \phi_1(x) - \cosech(\tau_0)\Delta \phi_0(x) \eqqcolon \psi_1(x),
\end{cases}
\end{equation}
where we used \eqref{hyperbvacuumwaveeqn2} to write
\begin{equation}
\partial_{\tau} \hat O \phi = \cosech^2(\tau)\cotanh(\tau)\partial_t \phi - \cosech(\tau) \Delta \phi.
\end{equation}
The spherical means solution is then
\begin{empheq}[box=\fbox]{align}\label{desitterhyperbsolution}
\phi(\tau, x) &=  -\int_{\tau_0}^{\tau} \frac{\cosh(s-\tau_0)}{4\pi} \sinh(s) \int_{S^2} \psi_0 \big(\exp_x ((s-\tau_0)z) \big) \, dV_2(z) ds \nonumber \\
& \quad - \int_{\tau_0}^{\tau} \frac{\sinh(s-\tau_0)}{4\pi} \sinh(s) \int_{S^2} d\psi_0(c'_z(s-\tau_0)) \, dV_2(z)  ds \\
& \quad - \int_{\tau_0}^{\tau} \frac{\sinh(s-\tau_0)}{4\pi} \sinh(s)\int_{S^2} \psi_1 \big(\exp_x ((s-\tau_0)z) \big) \, dV_2(z) ds \nonumber \\
& \quad + \phi_0 (x). \nonumber
\end{empheq}
Repeating the same steps as was done for the anti-de Sitter solution, one can show that when $\phi_1 \equiv 0$ we have
\begin{align}
\phi(\tau, x) & = \frac{\sinh(\tau)}{4 \pi \sinh(\tau_0) \sinh(\tau - \tau_0)} \int_{\partial B_{\tau-\tau_0} (x)} \nabla \phi_0(y) \cdot \nabla r(y) dV_2(y) \nonumber \\
& + \frac{1}{4 \pi \sinh^2(\tau - \tau_0)} \int_{\partial B_{\tau-\tau_0} (x)} \phi_0(y) dV_2(y).
\end{align}
Therefore, again $\phi$ satisfies Yagdjian's incomplete Huygens principle.

Similarly to what happens in the flat de Sitter universe, the time derivative of the solution decays exponentially as $t \to +\infty$,
\begin{equation}
|\partial_t \phi| \lesssim e^{-2t}.
\end{equation}
Indeed, if we take the partial derivative with respect to $\tau$ of both sides of \eqref{desitterhyperbsolution}, we have, as  $\tau \to 0$,
\begin{equation}\label{relationdesitter}
|\sinh(t) \partial_t \phi| =  |\partial_{\tau} \phi| \lesssim - \sinh{\tau} = \frac1{\sinh(t)},
\end{equation}
where we are assuming that the $L^\infty$ norms of $\psi_0$, $d\psi_0$ and $\psi_1$ are finite.
%
%
\subsection{Milne universe}
\label{subsection3.6}
We will now work with the scale factor
\begin{equation}
a(\tau ) = e^{\tau}.
\end{equation}
This model represents the Milne universe, as we can see by writing the scale factor as a function of the physical variable $t$:
\begin{equation}
t=\int e^{\tau} d\tau=e^{\tau},
\end{equation}
whence
\begin{equation}
a(t) = t
\end{equation}
for $t \in (0, +\infty)$.

With this choice, our wave equation \eqref{wavephi} becomes
\begin{equation}\label{milnewaveeqn}
\partial_{\tau}^2 \phi + 2 \partial_{\tau} \phi = \Delta \phi.
\end{equation}
To satisfy the constraint equations \eqref{def_fg_hyp}, we can pick $\alpha = \beta=1$ and $\kappa = 0$. In this case, the operator is 
\begin{equation}
\hat O \phi = e^{\tau} \phi.
\end{equation}
To get an explicit expression for solutions we consider the following initial data for $\hat O \phi$ at the initial conformal time $\tau=\tau_0$:
\begin{equation}\label{milneinitdata}
\begin{cases}
\hat O \phi(\tau_0, x) = e^{\tau_0} \phi_0(x) \eqqcolon \psi_0(x), \\
\partial_{\tau} \hat O \phi(\tau_0, x) = e^{\tau_0} \phi_0(x) + e^{2\tau_0} \phi_1(x) \eqqcolon \psi_1(x).
\end{cases}
\end{equation}
The solution in terms of spherical means is given by \eqref{hyperbolicsphericalmean2}:
\begin{empheq}[box=\fbox]{align} \label{milnesphericalmeans}
e^{\tau} \phi(\tau, x) &= \frac{\cosh(\tau-\tau_0)}{4\pi} \int_{S^2} \psi_0 \big (\exp_x((\tau-\tau_0)z) \big) \, dV_2(z)\, + \nonumber \\ 
&  + \frac{\sinh(\tau-\tau_0)}{4\pi} \int_{S^2} d \psi_0(c'_z(\tau-\tau_0)) \, dV_2(z)\, \\
&  + \frac{\sinh(\tau-\tau_0)}{4\pi} \int_{S^2} \psi_1 \big ( \exp_x((\tau-\tau_0)z) \big) \, dV_2(z). \nonumber
\end{empheq}
As could be expected from the fact that the Milne universe is a globally hyperbolic region of Minkowski's spacetime, $\phi$ satisfies the strong Huygens principle. Notice that the decay of the solution as $\tau \to +\infty$ is the same as that of the solution \eqref{hyperbradiationsphericalmeans} if we assume that the appropriate norms are finite. Thus, as $t \to +\infty$ we have
\begin{equation}
|\phi| \lesssim e^{-2\tau} = t^{-2}
\end{equation}
This decay rate can be understood by viewing the Milne universe as the region $t>|x|$ of Minkowski's spacetime and recalling the well-known estimate
\begin{equation}
|\phi(x,t)| \lesssim \frac1{(1+t+|x|)\sqrt{1+|t-|x||}}
\end{equation}
for solutions of the wave equation with compactly supported initial data (see for instance \cite{Sogge}).
%
%
\section{Spherical case}
\label{section4}
Finally, we consider the case $K=1$. The constraint equations \eqref{mustbe0} now give
\begin{equation}\label{operator_spherical}
\begin{cases}
f(\tau) = \kappa a, \\
g(\tau) = \alpha \cos(\tau) + \beta \sin(\tau), \\
\displaystyle 2 \frac{a^{\prime}}{a} g-2 g^{\prime}-f^{\prime \prime}+2\left(\frac{a^{\prime}}{a}\right)^{\prime} f - f = 0.
\end{cases}
\end{equation}
Again these conditions are satisfied for certain choices of $a(\tau)$, for which it is possible to obtain an explicit expression for $\phi$ from the expression of $\hat O \phi$ in terms of spherical means.
%
%
\subsection{Spherical means formula for $K=1$}
{\label{subsection4.1}
Let us find the solution of the conformally invariant wave equation \eqref{waveO} in the Einstein universe, that is, in the spacetime given by
\begin{equation}
g = -dt^2+d \varSigma_3^{\, 2},
\end{equation}
where $d\varSigma_3^{\, 2}$ is the line element for the $3$-sphere $S^3$. We therefore consider the Cauchy problem
\begin{equation} \label{wavespherical}
\begin{cases}
\partial_t^2 \phi = \Delta\phi-\phi, \\
\phi(0, x) = \phi_0(x), \\
\partial_t \phi(0, x) = \phi_1(x),
\end{cases}
\end{equation}
where $\Delta$ is the Laplacian operator in $S^3$. Following the exact same steps as in the hyperbolic case, but replacing the hyperbolic functions by their trigonometric counterparts, we arrive at the Kirchhoff-like formula
\begin{empheq}[box=\fbox]{align}  \label{sphericalsphericalmean2}
\phi(t, x) &= \frac{\cos(t)}{4 \pi} \int_{S^2} \phi_0 \big(\exp_x(tz) \big)\, dV_2(z) \nonumber \\
& + \frac{\sin(t)}{4 \pi} \int_{S^2} d\phi_0(\dot c_z(t)) \, dV_2(z)  \\
& + \frac{\sin(t)}{4 \pi} \int_{S^2} \phi_1 \big(\exp_x(tz) \big) \, dV_2(z), \nonumber
\end{empheq}
where again $c_z(t)=\exp_x(tz)$. It is interesting to note that $\phi$ is periodic in $t$ with period $2\pi$.
%
%
\subsection{Dust-filled spherical universe}
\label{subsection4.2}
Let us consider the case
\begin{equation}
a(\tau) = 1 -\cos(\tau)
\end{equation}
with $\tau \in (0, 2\pi)$. This corresponds to the dust-filled spherical model, since $a(\tau)$ solves Friedmann's equations for $K=1$, $\Lambda=0$ and $w = 0$ (see appendix~\ref{appendixA}).
With this choice of the scale factor, the wave equation \eqref{wavephi} becomes
\begin{equation}\label{waveeqn_sphericaldust}
\partial_{\tau}^2 \phi + \frac{2\sin(\tau)}{1-\cos(\tau)}\partial_{\tau} \phi = \Delta \phi.
\end{equation}
The constraint equations \eqref{operator_spherical} are satisfied if we choose $\kappa = 1$, $\alpha = 0$ and $\beta = \frac{3}{2}$, that is, if we choose
\begin{equation}
\hat O \phi =  (1-\cos(\tau)) \partial_{\tau}\phi + \frac{3}{2} \sin(\tau)\phi,
\end{equation}
which can also be written as
\begin{equation}\label{sphericaldust_tmpexpression}
\hat O \phi = \frac{2}{\sin\left( \frac{\tau}{2} \right)} \partial_{\tau} \left( \sin^3 \left(\frac{\tau}{2} \right) \phi \right),
\end{equation}
as remarked in \cite{KlainermanSarnak}. The initial data for $\hat O \phi$ at the initial conformal time $\tau=\tau_0$ is given by
\begin{equation} \label{initialdatasphericaldust}
\begin{cases}
\hat O \phi(\tau_0, x) = \frac{3}{2}\sin(\tau_0)\phi_0(x) + (1-\cos(\tau_0)) \phi_1(x) \eqqcolon \psi_0(x), \\
\partial_{\tau} \hat O \phi(\tau_0, x) = \frac{1}{2}\sin(\tau_0) (1-\cos(\tau_0))\phi_1(x) \\
\qquad \qquad \qquad +(1-\cos(\tau_0))\Delta \phi_0(x) + \frac{3}{2}\cos(\tau_0) \phi_0(x) \eqqcolon \psi_1(x),
\end{cases}
\end{equation}
where we used \eqref{waveeqn_sphericaldust} to write
\begin{equation}
\partial_{\tau} \hat O \phi = \frac{1}{2} \sin(\tau) (1-\cos(\tau)) \partial_t \phi + (1-\cos(\tau)) \Delta \phi + \frac{3}{2} \cos(\tau) \phi.
\end{equation}
Finally, the spherical means formula \eqref{sphericalsphericalmean2} and \eqref{sphericaldust_tmpexpression} give:
\begin{empheq}[box=\fbox]{align} \label{sphericalmeans_sphericaldust} 
\sin^3&\left(\frac{\tau}{2} \right) \phi(\tau, x) = \sin^3\left(\frac{\tau_0}{2}\right)\phi_0(x) \nonumber \\
& + \int_{\tau_0}^{\tau}  \frac{\cos(s-\tau_0)}{8 \pi} \sin\left( \frac{s}{2} \right) \int_{S^2} \psi_0 \big(\exp_x((s-\tau_0)z) \big)\, dV_2(z) ds \nonumber \\
& + \int_{\tau_0}^{\tau} \frac{\sin(s-\tau_0)}{8 \pi} \sin\left( \frac{s}{2} \right) \int_{S^2} d\psi_0(c_z'(s-\tau_0)) \, dV_2(z) ds \\
& + \int_{\tau_0}^{\tau} \frac{\sin(s-\tau_0)}{8 \pi} \sin\left( \frac{s}{2} \right) \int_{S^2} \psi_1 \big(\exp_x((s-\tau_0)z) \big) \, dV_2(z) ds. \nonumber
\end{empheq}
In general, these solutions will diverge as $\tau \to 0$ and as $\tau \to 2\pi$. As in the flat case, there exist solutions which have a well-defined limit at the Big Bang: if we take the limit $\tau_0 \to 0$ while keeping $\phi_0$, $\phi_1$ and $\tau$ fixed, we obtain from \eqref{initialdatasphericaldust} and \eqref{sphericalmeans_sphericaldust}  the limit solution
\begin{equation}
\phi\left (\tau, x \right) = \frac{3}{16\pi\sin^3\left( \frac{\tau}{2} \right)} \int_0^{\tau} \sin(s) \sin\left( \frac{s}{2} \right) \int_{S^2} \phi_0(\exp_x(zs)) dV_2(z) ds.
\end{equation}
Using L'H\^{o}pital's rule, one can easily check that
\begin{equation}
\lim_{\tau \to 0} \phi(\tau, x) =  \phi_0(x).
\end{equation}
Similarly, we can find solutions with a well-defined limit at the Big Crunch by taking the limit $\tau_0 \to 2\pi$ while keeping $\phi_0$, $\phi_1$ and $\tau$ fixed. In this case we obtain the limit solution
\begin{equation}
\phi\left (\tau, x \right) = \frac{3}{16\pi\sin^3\left( \frac{\tau}{2} \right)} \int_{2\pi}^{\tau} \sin(s) \sin\left( \frac{s}{2} \right) \int_{S^2} \phi_0(\exp_x(zs)) dV_2(z) ds,
\end{equation}
which, again using L'H\^{o}pital's rule, satisfies
\begin{equation}
\lim_{\tau \to 2\pi} \phi(\tau, x) =  \phi_0(x).
\end{equation}
Interestingly, if $\phi_0$ satisfies
\begin{equation}
\int_{0}^{2\pi} \sin(s) \sin\left( \frac{s}{2} \right) \int_{S^2} \phi_0(\exp_x(zs)) dV_2(z) ds = 0
\end{equation}
for all $x \in S^3$ then the two limit solutions coincide, yielding a solution with the same limit at the Big Bang and at the Big Crunch (this includes, of course, the constant solutions).
%
%
\subsection{Radiation-filled spherical universe}
\label{subsection4.3}
Next, let us take the scale factor
\begin{equation}
a(\tau)= \sin(\tau)
\end{equation}
with $\tau \in (0,\pi)$. This corresponds to a radiation-dominated spherical universe, since $a(\tau)=\sin(\tau)$ solves the Friedmann equations for $K=1$, $\Lambda=0$ and $w = \frac{1}{3}$ (see appendix~\ref{appendixA}).
With this choice of the scale factor, our wave equation \eqref{wavephi} becomes
\begin{equation}\label{waveeqn_sphericalrad}
\partial_{\tau}^2 \phi + 2 \cotan(\tau) \partial_{\tau} \phi = \Delta \phi.
\end{equation}
In this case, the constraint equations \eqref{def_fg_hyp} are satisfied if we choose $\alpha = \kappa = 0$ and $\beta = 1$, that is, if we choose
\begin{equation}
\hat O \phi = \sin(\tau) \phi.
\end{equation}
Again, this operator can be taken as multiplicative because in this universe the wave equation coincides with the conformally invariant wave equation.

The initial data for $\hat O \phi$ at the initial conformal time $\tau = \tau_0$ is given by:
\begin{equation} \label{initialdatasphericalradiation}
\begin{cases}
\hat O \phi(\tau_0, x) = \sin(\tau_0) \phi_0(x) \eqqcolon \psi_0(x), \\
\partial_{\tau} \hat O \phi(\tau_0, x) = \cos(\tau_0) \phi_0(x) + \sin^2(\tau_0) \phi_1(x) \eqqcolon \psi_1(x).
\end{cases}
\end{equation}
The spherical means solution is then given by \eqref{sphericalsphericalmean2}:
\begin{empheq}[box=\fbox]{align} \label{sphericalmeans_sphericalradiation}
\phi(\tau, x) &= \frac{\cos(\tau-\tau_0)}{4 \pi \sin(\tau)} \int_{S^2} \psi_0 \big(\exp_x((\tau-\tau_0)z) \big)\, dV_2(z) \nonumber \\
& + \frac{\sin(\tau-\tau_0)}{4 \pi \sin(\tau)} \int_{S^2} d\psi_0 \left(\dot c_z (\tau-\tau_0) \right) \, dV_2(z) \\
& + \frac{\sin(\tau-\tau_0)}{4 \pi \sin(\tau)} \int_{S^2} \psi_1 \big(\exp_x((\tau-\tau_0)z) \big) \, dV_2(z).\nonumber
\end{empheq}
As in the flat and hyperbolic cases, $\phi$ satisfies the strong Huygens principle.

Again, these solutions will in general diverge as $\tau \to 0$ and as $\tau \to \pi$, but there exist solutions which have a well-defined limit at the Big Bang: if we take the limit $\tau_0 \to 0$ while keeping $\phi_0$, $\phi_1$ and $\tau$ fixed, we obtain from \eqref{initialdatasphericalradiation} and \eqref{sphericalmeans_sphericalradiation} the limit solution
\begin{equation}
\phi(\tau, x)=\frac{1}{4\pi} \int_{S^2}\phi_0 \big(\exp_x(\tau z) \big) dV_2(z),
\end{equation}
which clearly converges to  $\phi_0(x)$ as $\tau \to 0$. Similarly, we can find solutions with a well-defined limit at the Big Crunch by taking the limit $\tau_0 \to \pi$ while keeping $\phi_0$, $\phi_1$ and $\tau$ fixed. In this case we obtain the limit solution
\begin{equation}
\phi(\tau, x)=\frac{1}{4\pi} \int_{S^2}\phi_0 \big(\exp_x((\tau-\pi) z) \big) dV_2(z),
\end{equation}
which converges to $\phi_0(x)$ as $\tau \to \pi$. Interestingly, if $\phi_0$ is an even function on the sphere $S^3$, then the two limit solutions coincide, yielding a solution with the same limit at the Big Bang and at the Big Crunch (this includes, of course, the constant solutions).
%
%
\subsection{Spherical de Sitter universe}
\label{subsection4.4}
Finally, we consider the case
\begin{equation}
a(\tau ) = \sec(\tau)
\end{equation}
with $\tau \in \left (-\frac{\pi}{2}, \frac{\pi}{2} \right)$. This corresponds to the spherical (full) de Sitter universe, since $a(\tau)= \sec(\tau)$ is a solution for the Friedmann equations for $K=1$, $\Lambda=3$ and $\rho_0= 0$  (see appendix~\ref{appendixA}). We can obtain the scale factor as a function of the physical time $t$ by noticing that
\begin{equation}
t= \int \frac{d\tau}{\cos(\tau)} = 2 \arctanh \left( \tan \left ( \frac{\tau}{2} \right) \right),
\end{equation}
and so
\begin{equation}
a(t) = \sec \left( 2 \arctan \left( \tanh\left( \frac{t}{2} \right) \right) \right) = \cosh (t).
\end{equation}

With this choice of the scale factor, \eqref{wavephi} becomes
\begin{equation}\label{sphericalbvacuumwaveeqn}
\partial_{\tau}^2 \phi + 2 \tan(\tau) \partial_{\tau} \phi = \Delta \phi.
\end{equation}
We can satisfy the constraint equations \eqref{operator_spherical} by choosing $\alpha = \beta = 0$ and $\kappa = 1$, so that we have
\begin{equation}
\hat O \phi = \sec(\tau) \partial_{\tau} \phi.
\end{equation}
 As before, we can express the solution in terms of spherical means. We have the following initial data for $\hat O \phi$ at the initial conformal time $\tau = \tau_0$:
\begin{equation}\label{sphericalbvacuuminitialdata1}
\begin{cases}
\hat O \phi(\tau_0, x) = \sec^2(\tau_0) \phi_1(x) \eqqcolon \psi_0(x), \\
\partial_{\tau} \hat O \phi(\tau_0, x) =  \sec(\tau_0) \Delta \phi_0(x) - \sec^2(\tau_0) \tan(\tau_0) \phi_1(x) \eqqcolon \psi_1(x),
\end{cases}
\end{equation}
where we used  \eqref{sphericalbvacuumwaveeqn} to write
\begin{equation}
\partial_{\tau} \hat O \phi = \sec(\tau) \Delta \phi - \sec^2(\tau)\tan(\tau)\partial_t \phi.
\end{equation}
The spherical means solution is then
\begin{empheq}[box=\fbox]{align} \label{solutionsphericaldS}
\phi(\tau, x) &=  \int_{\tau_0}^{\tau} \frac{\cos(s-\tau_0)}{4\pi} \cos(s) \int_{S^2} \psi_0 \big(\exp_x ((s-\tau_0)z) \big) \, dV_2(z) ds  \nonumber \\
&  + \int_{\tau_0}^{\tau} \frac{\sin(s-\tau_0)}{4\pi} \cos(s) \int_{S^2} d\psi_0(c'_z(s-\tau_0)) \, dV_2(z)  ds \\
&  + \int_{\tau_0}^{\tau} \frac{\sin(s-\tau_0)}{4\pi} \cos(s)\int_{S^2} \psi_1 \big (\exp_x ((s-\tau_0)z) \big ) \, dV_2(z) ds \nonumber \\
&  + \phi_0(x). \nonumber
\end{empheq}
In the particular case when $\phi_1 \equiv 0$ we obtain
\begin{align}
\phi(\tau, x) & = \phi_0(x) + \int_{0}^{\tau-\tau_0} \frac{\sin(r)}{4\pi} \cos(r+\tau_0)\int_{S^2} \sec(\tau_0) \Delta\phi_0 \big (\exp_x (rz) \big ) \, dV_2(z) dr \nonumber \\
& = \phi_0(x) + \frac{1}{4\pi} \int_{B_{\tau-\tau_0} (x)} \left( \cotan(r) - \tan(\tau_0) \right) \Delta\phi_0(y) dV_3(y),
\end{align}
where we defined $r(y)=\dists(y, x)$. Using the divergence theorem and Green's identities, together with $\Delta \cotan(r) = 0$, this expression can be rewritten as
\begin{align}
\phi(\tau, x) & = \frac{\cos(\tau)}{4 \pi \cos(\tau_0) \sin(\tau - \tau_0)} \int_{\partial B_{\tau-\tau_0} (x)} \nabla \phi_0(y) \cdot \nabla r(y) dV_2(y) \nonumber \\
& + \frac{1}{4 \pi \sin^2(\tau - \tau_0)} \int_{\partial B_{\tau-\tau_0} (x)} \phi_0(y) dV_2(y).
\end{align}
Therefore, $\phi(\tau, x)$ depends only on the values of $\phi_0$ and $\nabla \phi_0$ on $\partial B_{\tau-\tau_0} (x)$, that is, $\phi$ satisfies Yagdjian's incomplete Huygens principle.

As expected by analogy with the flat and hyperbolic de Sitter universes, the time derivative of the solution decays exponentially as $t \to +\infty$,
\begin{equation}
|\partial_t \phi| \lesssim e^{-2t}.
\end{equation}
Indeed, if we take the partial derivative with respect to $\tau$ of both sides of \eqref{solutionsphericaldS}, we have, as  $\tau \to \frac{\pi}2$,
\begin{equation}
|\cosh(t) \partial_t \phi| =  |\partial_{\tau} \phi| \lesssim \cos{\tau} = \frac1{\cosh(t)}.
\end{equation}
Similarly, one can show that the time derivative of the solution decays exponentially as $t \to -\infty$,
\begin{equation}
|\partial_t \phi| \lesssim e^{2t}.
\end{equation}

%
%
\section{Decay rates in flat and hyperbolic FLRW universes}\label{section6}
Let us consider the family of flat FLRW universes described by the scale factor $a(t)=t^p$, with $p \geq 0$. In this context, we obtained the Minkowski spacetime decay rate of $t^{-1}$ for the wave equation when $p=\frac{2}{3}$ (dust-filled universe, see \eqref{flatdecaydust}) and when $p=\frac{1}{2}$ (radiation-filled universe, see \eqref{flatdecayradiation}). The independence of these results on $p$ (as well as the numerical results in \cite{RossettiVinuales}) suggests a wider range of validity for this decay. 

\begin{Conjecture}[Decay in the flat case]\label{conjecture-1}
Let $\phi$ be a solution of the wave equation in a flat FLRW universe with scale factor $a(t)=t^p$, where $p \ge 0$. Assume that the initial data $\phi_0(x)\coloneqq \phi(t_0, x)$ and $\phi_1(x)\coloneqq\partial_t \phi(t_0, x)$ is sufficiently regular and belongs to appropriate Sobolev spaces. Then
\begin{equation}
\| \phi(t, \cdot) \|_{L^{\infty}(\mathbb{R}^3)} \lesssim \begin{cases}
(1+t)^{-1}, &0 \le p \le \frac23,\\
(1+t)^{-3(1-p)}, &\frac23 \le p < 1,
\end{cases}
\end{equation}
or, equivalently:
\begin{equation}
\| \phi(t, \cdot) \|_{L^{\infty}(\mathbb{R}^3)} \lesssim \begin{cases}
(1+t)^{-1}, & w \ge 0,\\
(1+t)^{-\frac{3w+1}{w+1}}, &-\frac13 < w \le 0.
\end{cases}
\end{equation}
Moreover, we have 
\[
\|\phi(t, \cdot)\|_{L^{\infty}(\bbR^3)} \lesssim (\log t)^{-\frac32} \quad \text{ for } \,\,  p=1 \,\, \left(\text{i.e. }w = -\frac13 \right) ,
\]
and there is no decay for $p > 1$ (i.e. $-1 < w < -\frac13$).
\end{Conjecture}

The absence of decay for $p>1$ (but {\em not} for $p=1$) follows from the results in \cite{CostaNatarioOliveira}). For $0 \le p \le 1$ we have the following result:

\begin{Thm}
Conjecture~\ref{conjecture-1} holds for $0 \le p \le 1$ (i.e. $w \ge -\frac13$).
\end{Thm}
\begin{proof} 
The case $p=0$ corresponds to the wave equation in Minkowski's spacetime. For $0 < p <1$, the original wave equation \eqref{wavephi} (written in conformal time) can be seen as a damped wave equation in Minkowski's spacetime:
\begin{equation}
\square \phi - \frac{\nu}{\tau} \partial_{\tau}\phi = 0,
\end{equation}
where $\Box = -\partial^2_{\tau} + \Delta$ and 
\begin{equation}
\nu \coloneqq \frac{2p}{1-p}.
\end{equation}
This equation was studied in \cite{Wirth}, theorem 3.5 (which can be extended to the $L^1$--$L^{\infty}$ case, as stressed in remark 3.3), where it was proved that 
\[
\| \phi(\tau, \cdot) \|_{L^{\infty}(\mathbb{R}^3)} \lesssim \begin{cases}
(1 + \tau)^{-1-\frac{\nu}{2}} \sim (1+t)^{-1}, &0 < p \le \frac23,\\
(1+\tau)^{-3} \sim (1+t)^{-3(1-p)}, &\frac23 \le p < 1.
\end{cases} 
\]
When $p=1$, the respective damped equation has constant coefficients. By a simple rescaling: $\tilde \phi(\tau, x) \coloneqq \phi\left(\frac{\tau}{2}, \frac{x}{2}\right)$, we have
\[
\square\tilde \phi -\partial_{\tau} \tilde \phi = 0.
\]
The diffusive structure of the above PDE was explicitly described in \cite{Nishihara}, where the decay
\[
\|\tilde \phi\|_{L^{\infty}(\bbR^3)} \lesssim \tau^{-\frac32} \sim (\log t)^{-\frac32}
\]
was established.
\end{proof}
A related conjecture for the hyperbolic case is motivated by the decay rates obtained in sections~\ref{subsection3.2} and \ref{subsection3.3}, and by the numerical results in \cite{RossettiVinuales}.

\begin{Conjecture}[Decay in the hyperbolic case]
Let $\phi$ be a solution of the wave equation in a hyperbolic FLRW universe, with scale factor satisfying the Friedmann equations \eqref{friedmann1} with zero cosmological constant and equation of state \eqref{eqnofstate}. Assume that the initial data $\phi_0(x)\coloneqq \phi(t_0, x)$ and $\phi_1(x)\coloneqq\partial_t \phi(t_0, x)$ is sufficiently regular and belongs to appropriate Sobolev spaces. Then
\begin{equation}
\| \phi(t, \cdot) \|_{L^{\infty}(\mathbb{H}^3)} \lesssim  \begin{cases}
(1+t)^{-2}, & w \ge \frac13,\\
(1+t)^{-\frac{3(w+1)}{2}}, &0 \le w \le \frac13.
\end{cases}
\end{equation}
Moreover, the decay is slower than $(1+t)^{-\frac{3(w+1)}{2}}$ for $-\frac13 < w < 0$, and there is no decay for $-1 \le w < -\frac13$. 
\end{Conjecture}
%
%
\section*{Acknowledgements}
We thank Jo\~ao Costa and  Alex Va\~{n}\'{o}-Vi\~{n}uales for many useful discussions. This work was partially supported by FCT/Portugal through CAMGSD, IST-ID, projects UIDB/04459/2020 and UIDP/04459/2020. Flavio Rossetti was supported by FCT/Portugal through the PhD scholarship UI/BD/152068/2021.
%
%
\appendix
%
%
\section{The Friedmann equations}\label{appendixA}
Meaningful choices of the scale factor $a(t)$ in the FLRW metric \eqref{flrwmetric3} are obtained by solving the Friedmann equations, which result from the Einstein equations with an ideal fluid source.
In three spatial dimensions, these are given by (see for instance \cite{ChenGibbons}):
\begin{equation} \label{friedmann1}
\left ( \frac{\dot a}{a} \right)^2 = \frac{8 \pi}{3} \rho_m + \frac{\Lambda}3 - \frac{K}{a^2} \qquad \text{ and} \qquad \frac{\ddot a}{a} = -\frac{4 \pi}{3}(\rho_m + 3p_m) + \frac{\Lambda}3,
\end{equation}
where $K$ is the curvature of the spatial sections, $\rho_m$ is the fluid's energy density, $p_m$ is the fluid's pressure, and $\Lambda$ is the cosmological constant. If we assume that the fluid satisfies the linear equation of state
\begin{equation} \label{eqnofstate}
p_m = w \rho_m,
\end{equation} 
where $w$ is a constant (the square of the fluid's speed of sound), then \eqref{friedmann1} can be rewritten as
\begin{equation}\label{friedmann1finalconf}
\left ( \frac{ a'}{a^2} \right)^2 = \frac{8 \pi}{3} \rho_0 a^{-3(1+w)} + \frac{\Lambda}{3} - \frac{K}{a^2},
\end{equation}
where $\rho_0 \geq 0$ is an integration constant and the derivatives are now taken with respect to the conformal time.
Interesting values of $w$ are given in table~\ref{aa1}.

\setlength{\tabcolsep}{40pt}
\begin{table}[H]
{\rowcolors{2}{white}{gray!10}
\begin{tabular}{ c | c } \toprule
 $w$ & fluid type \\[.4em]
 \hline
 $w<0$ & unphysical (imaginary speed of sound) \\
 $w=0$ & dust \\  
 $w=\frac{1}{3}$ & radiation \\
 $w=1$ & stiff fluid \\
 $w>1$ & unphysical (superluminal speed of sound) \\
\bottomrule
\end{tabular}
}
\caption{Fluid types for different values of $w$}\label{aa1}
\end{table}
\bigskip

If $\rho_0>0$, $\Lambda=0$, $K=0$ and $w \neq -\frac13$, the solutions of \eqref{friedmann1finalconf} are given by
\begin{equation} \label{flatexactsol}
a(\tau) = \tau^{\frac{2}{1+3w}}
\end{equation}
for an appropriate choice of $\rho_0$ (amounting to a choice of units). In this work we consider the following scale factors:
\begin{equation}
a(\tau) = \begin{cases}
\tau^2 &\text{(dust),}\\
\tau & \text{(radiation)},\\
\tau^{-1} &\text{(de Sitter)}. 
\end{cases}
\end{equation}
The de Sitter universe corresponds to $\rho_0=0$ and $\Lambda > 0$ (here we chose units such that $\Lambda = 3$).

If $\rho_0>0$, $\Lambda=0$, $K=-1$ and $w \neq -\frac13$, the solutions of \eqref{friedmann1finalconf} are given by
\begin{equation} \label{hyperbexactsol}
a(\tau) = 
\left [ \sinh\left(\frac{1}{2}(3w+1) \tau \right) \right]^{\frac{2}{3w+1}}
\end{equation}
for an appropriate choice of $\rho_0$. In particular, we consider the following scale factors:
\begin{equation}
a(\tau) =
\begin{cases}
\cosh(\tau)-1 & \text{(dust),}\\
\sinh(\tau) & \text{(radiation),}\\
\sech(\tau) &\text{(anti-de Sitter)}, \\
-\cosech(\tau) &\text{(de Sitter)}, \\
e^{\tau} &\text{(Milne universe)}.
\end{cases}
\end{equation}
The anti-de Sitter, de Sitter and Milne universes correspond to the cases where $\rho_0 = 0$ and $\Lambda = -3$, $\Lambda = 3$ and $\Lambda = 0$, respectively.

Finally, if $\rho_0>0$, $\Lambda=0$, $K=1$ and $w \neq -\frac13$, the solutions of \eqref{friedmann1finalconf} are given by
\begin{equation} 
a(\tau) = 
\left [ \sin\left(\frac{1}{2}(3w+1) \tau \right) \right]^{\frac{2}{3w+1}}
\end{equation}
for an appropriate choice of $\rho_0$. In particular, we consider:
\begin{equation}
a(\tau) = 
\begin{cases}
1-\cos(\tau) &\text{(dust)}, \\
\sin(\tau) &\text{(radiation)}, \\
\sec(\tau) &\text{(de Sitter)}.
\end{cases}
\end{equation}
Again, the de Sitter universe corresponds to $\rho_0 = 0$ and $\Lambda = 3$.
%
%
%
\section{Decay for the dust-filled hyperbolic universe}\label{appendixB}
Let us show that the decay \eqref{hypdecaydust} is optimal by fixing a point $\bar x \in \mathbb{H}^3$ and choosing initial data such that $|\phi(\tau, \bar x)| \sim e^{-\frac{3}{2}\tau}$. This can be done as follows: given $\tau_0 > 0$, we choose the functions $\phi_0$ and $\phi_1$ such that $\psi_0 \equiv 0$. From \eqref{initialdatahyperbdust}, this implies that
\begin{equation}
\psi_1(x) = A \phi_0(x) + B \Delta \phi_0(x)
\end{equation}
with $A = \frac{3}{2} \left( \cosh(\tau_0) - \frac{1}{2} \sinh^2(\tau_0) \right) $ and $B = \cosh(\tau_0)-1$. 
Next, we choose the radially symmetric function $\phi_0(x) = \varphi(\disth(x, \bar x))$, where
\begin{equation}\label{phi0def}
\varphi(r) \coloneqq \begin{cases}
N r^3 \exp \left( \frac{1}{r-1} \right), &0 \le r < 1, \\
0, & r \ge 1,
\end{cases}
\end{equation}
and $N > 0$ is a free parameter. Note that $\varphi$ is smooth in $(0,+\infty)$, and moreover that its behavior at $r=0$ implies that $\phi_0$ is of class $C^2$.
Plugging this choice of $\phi_0$ into the right-hand side of the spherical means solution \eqref{hyperbdustsolution} evaluated at $\bar x$, we obtain
\begin{align}\label{fdef}
&\int_{\tau_0}^{\tau} \frac{\sinh(s-\tau_0)}{8 \pi } \sinh\left (\frac{s}{2} \right ) \int_{S^2} \psi_1 \big(\exp_{\bar x}((s-\tau_0)z) \big) \, dV_2(z)ds  \\
&= \frac{1}{2} \int_0^{\tau-\tau_0} \sinh(r) \sinh \left( \frac{r+\tau_0}{2} \right) \big (A \varphi(r)+B\varphi''(r)+2B \cotanh(r) \varphi'(r) \big) dr, \nonumber
\end{align}
where we used \eqref{laplacianhyperb} and the substitution $r=s-\tau_0$. By \eqref{phi0def}, the above integral is constant for $\tau > \tau_0+1$. We can choose $\tau_0$ such that this constant is non-zero\footnote{For $\tau_0=1$,  $\tau=2$ and $N=50$, the evaluation of the integral in \eqref{fdef} using Mathematica gives 0.151503, and this value grows linearly with $N$.}, and so, after this choice of $\phi_0$ and $\phi_1$, \eqref{hyperbdustsolution} implies
\begin{equation}
\phi(\tau, \bar x)  = C \sinh^{-3} \left( \frac{\tau}{2} \right )  \sim e^{-\frac{3}{2} \tau} \text{ as } \tau \to +\infty.
\end{equation}
Since $C^2$ functions on a manifold can be approximated by smooth functions (Whitney approximation theorem, see e.g. \cite{Lee}), we can choose smooth initial data that approximates the chosen $\phi_0$ and extend this construction to smooth initial data. The above reasoning disproves the decay rate of $e^{-2\tau}$ in the conformal time variable stated by Abbasi and Craig in \cite{AbbasiCraig}. Indeed, the inequality proved in theorem 3.1 of their article seems to hold specifically for the term they choose, but not for the last terms in their solutions (3.9) and (3.10). 
%


\begin{thebibliography}{99}


\bibitem{AbbasiCraig}
B.~Abbasi and W.~Craig, \emph{On the initial value problem for the wave equation in Friedmann-Robertson-Walker spacetimes}, Proc.\ R.\ Soc.\ A {\bf 470} (2014) 20140361.

\bibitem{ChenGibbons}
S. ~Chen, G. ~Gibbons, Y. ~Li and Y. ~Yang, \emph{Friedmann's Equations in All Dimensions and Chebyshev's Theorem}, JCAP 1412 (2014) 035.

\bibitem{CostaNatarioOliveira}
J.~Costa, J.~Nat\'{a}rio, and P.~Oliveira, \emph{Decay of solutions of the wave equation in expanding cosmological spacetimes}, J. \ Hyperbolic \ Differ. \ Equ. {\bf 16} (2019) 35-58.

\bibitem{CraigReddy}
W.~Craig and M.~Reddy, \emph{On the Initial Value Problem for the Electromagnetic Wave Equation in Friedmann-Robertson-Walker Space-times}, C.\ R.\ Math.\ Rep.\ Acad.\ Sci.\ Canada {\bf 41} (2019) 45-56.

\bibitem{DafermosRodnianski1}
M.~Dafermos and I.~Rodnianski, \emph{The wave equation on Schwarzschild-de Sitter spacetimes}, \href{https://arxiv.org/abs/0709.2766}{arXiv:0709.2766}.

\bibitem{Gajic}
D.~Gajic, \emph{Linear waves on constant radius limits of cosmological black hole spacetimes}, Adv.\ Theor.\ Math.\ Phys.\ {\bf 22} (2018) 919-1005.

\bibitem{GalstianKinoshitaYagdjian}
A.~Galstian, T.~Kinoshita and K.~Yagdjian, \emph{A note on wave equation in Einstein \& de Sitter spacetime}, J.\ Math.\ Phys.\ {\bf 51} (2010) 052501.

\bibitem{GalstianYagdjian}
A.~Galstian and K.~Yagdjian, \emph{Microlocal analysis for waves propagating in Einstein \& de Sitter spacetime}, Math.\ Phys.\ Anal.\ Geom.\ {\bf 17} (2014) 223-246.

\bibitem{GiraoNatarioSilva}
P.~Gir\~{a}o, J.~Nat\'{a}rio and J.~Silva, \emph{Solutions of the wave equation bounded at the Big Bang}, Class.\ Quantum Grav.\ {\bf 36} (2019) 075016.

\bibitem{Hormander}
L.~H\"{o}rmander, \emph{Lectures on Nonlinear Hyperbolic Differential Equations}, Springer (1997).

\bibitem{KlainermanSarnak}
S.~Klainerman and P.~Sarnak, \emph{Explicit solutions of $\Box u = 0$ on the Friedmann-Robertson-Walker space-times}, Ann.\ Inst.\ H.\ Poincar\'{e} Sect. A (N.S.) {\bf 35} (1981) 253-257.

\bibitem{KulczyckiMalec}
W.~Kulczycki and E.~Malec, \emph{Gravitational waves in Friedman-Lemaitre-Robertson-Walker cosmology, material perturbations and cosmological rotation, and the Huygens principle}, Class.\ Quantum Grav.\ {\bf 34} (2017) 135014.

\bibitem{Lee}
J.~Lee, \emph{Introduction to Riemannian Manifolds}, Springer (2018).

\bibitem{MalecWylezek}
E.~Malec and G.~Wyl\c{e}\.{z}ek, \emph{The Huygens principle and cosmological gravitational waves in the Regge-Wheeler gauge}, Class.\ Quantum Grav.\ {\bf 22} (2005) 3549.

\bibitem{Natario}
J.~Nat\'{a}rio, \emph{An Introduction to Mathematical Relativity}, Springer (2021).

\bibitem{NatarioSasane}
J.~Nat\'{a}rio and A.~Sasane, \emph{Decay of solutions to the Klein-Gordon equation on some expanding cosmological spacetimes},  Ann.\ Henri Poincar\'{e} {\bf 23} (2022), 2345-2389.

\bibitem{Nishihara}
K.~Nishihara, \emph{$L^p$--$L^q$ estimates of solutions to the damped wave equation in 3-dimensional space and their application}, Math. Z. {\bf 244} (2003), 631-649.

\bibitem{Ringstrom1}
H.~Ringstr\"{o}m, \emph{A unified approach to the Klein-Gordon equation on Bianchi backgrounds}, Commun.\ Math.\ Phys.\ {\bf 372} (2019) 599-656.

\bibitem{Ringstrom2}
H.~Ringstr\"{o}m, \emph{Linear systems of wave equations on cosmological backgrounds with convergent asymptotics}, Ast\'{e}risque {\bf 420} (2020), 1-526.

\bibitem{Ringstrom3}
H.~Ringstr\"{o}m, \emph{Wave equations on silent big bang backgrounds}, \href{https://arxiv.org/abs/2101.04939}{arXiv:2101.04939}.

\bibitem{Ringstrom4}
H.~Ringstr\"{o}m, \emph{On the geometry of silent and anisotropic big bang singularities}, \href{https://arxiv.org/abs/2101.04955}{arXiv:2101.04955}.

\bibitem{RossettiVinuales}
F.~Rossetti and A.~Va\~{n}\'{o}-Vi\~{n}uales, \emph{Decay of solutions of the wave equation in cosmological spacetimes -- a numerical analysis}, in preparation. 

\bibitem{Schlue}
V.~Schlue, \emph{Global results for linear waves on expanding Kerr and Schwarzschild de Sitter cosmologies}, Commun.\ Math.\ Phys.\ {\bf 334} (2015) 977-1023.

\bibitem{Sogge}
C.~Sogge, \emph{Lectures on non-linear wave equations}, International Press (2008).

\bibitem{StarkoCraig}
D. Starko and W.~Craig, \emph{The wave equation in Friedmann-Robertson-Walker space-times and asymptotics of the intensity and distance relationship of a localised source}, J.\ Math.\ Phys.\ {\bf 59} (2018) 042502.

\bibitem{Yagdjian}
K.~Yagdjian, \emph{Huygens' Principle for the Klein-Gordon equation in the de Sitter spacetime}, J.\ Math.\ Phys.\ {\bf 54} (2013) 091503.

\bibitem{YagdjianGalstian}
K.~Yagdjian and A.~Galstian, \emph{Fundamental solutions of the wave equation in Robertson-Walker spaces}, J.\ Math.\ Anal.\ Appl.\ {\bf 346} (2008) 501-520.

\bibitem{Wirth}
J.~Wirth, \emph{Solution representations for a wave equation with weak dissipation}, Math.\ Meth.\ Appl.\ Sci.\ {\bf 27} (2004) 101-124.

\end{thebibliography}
\end{document}